\documentclass{egpubl}

 \JournalSubmission    
 \electronicVersion 

\ifpdf \usepackage[pdftex]{graphicx} \pdfcompresslevel=9
\else \usepackage[dvips]{graphicx} \fi

\PrintedOrElectronic

\usepackage{t1enc,dfadobe}

\usepackage{egweblnk}
\usepackage{cite}

\usepackage{amsmath,amssymb,amsfonts}
\usepackage{overpic}
\usepackage{pgfplots}
\usepackage{pgf,tikz}
\usepackage{wrapfig}
\usepackage{nccmath}
\usepackage{MyMnSymbol} 

\newcommand{\M}{\mathcal{M}}
\newcommand{\N}{\mathcal{N}}

\renewcommand{\d}{\mathrm{d}}
\newcommand{\A}{\mathbf{A}}
\newcommand{\B}{\mathbf{B}}
\newcommand{\C}{\mathbf{C}}
\renewcommand{\S}{\mathbf{S}}
\newcommand{\W}{\mathbf{W}}
\renewcommand{\L}{\mathbf{L}}

\title[Functional Maps Representation on Product Manifolds]%
      {Functional Maps Representation on Product Manifolds}

\newcommand{\rev}[1]{{\color{black}{#1}}}
\newcommand{\final}[1]{{#1}}

\newlength\figureheight
\newlength\figurewidth

\newtheorem{theorem}{Theorem}

\author[E. Rodol\`{a} et al.]
{\parbox{\textwidth}{\centering
E. Rodol\`{a}$^{1}$ ~~~ Z. L{\"a}hner$^{2}$ ~~~ A. M. Bronstein$^{3,6}$ ~~~ M. M. Bronstein$^{4,5,6}$ ~~~ J. Solomon$^{7}$\\
rodola@di.uniroma1.it ~~
         laehner@in.tum.de ~~
         bron@cs.technion.ac.il ~~
	  bronstein@imperial.ac.uk ~~
         jsolomon@mit.edu
        }
        \\
{\parbox{\textwidth}{\centering
         $^1$Sapienza University of Rome ~~~~~
         $^2$TU Munich ~~~~~
         $^3$Technion ~~~~~
	$^4$USI Lugano ~~~~~
	$^5$Imperial College London ~~~~~
         $^6$Intel ~~~~~
         $^7$MIT
  }
}
}

\begin{document}


\maketitle

\begin{abstract}
We consider the tasks of representing, analyzing and manipulating maps between shapes. We model maps as densities over the product manifold of the input shapes; these densities can be treated as scalar functions and therefore are manipulable using the language of signal processing on manifolds. Being a manifold itself, the product space endows the set of maps with a geometry of its own, which we exploit to define map operations in the spectral domain; we also derive relationships with other existing representations (soft maps and functional maps). To apply these ideas in practice, we discretize product manifolds and their Laplace--Beltrami operators, and we introduce localized spectral analysis of the product manifold as a novel tool for map processing. Our framework applies to maps defined between and across 2D and 3D shapes without requiring special adjustment, and it can be implemented efficiently with simple operations on sparse matrices.

\vspace{1ex}

\end{abstract}

\section{Introduction}

3D acquisition continues to reach new levels of sophistication and is rapidly being incorporated into commercial products ranging from the Microsoft Kinect for gaming to LIDAR for autonomous cars and MRI for medical imaging. An essential building block for application design in many of these domains is fast and reliable recovery of 3D shape correspondences. This problem arises in applications as diverse as character animation, 3D avatars, pose and style transfer, or texture mapping, to mention a few.

A modern theme in shape correspondence involves the \emph{representation} of a map from one shape to another.  While the most obvious representation maintains pairs of source and target points, this is by no means the only option.  \
Our paper is mainly related to two frameworks developed for establishing correspondence between shapes: optimization on {\em product manifolds} and {\em functional maps}. 




The first class of methods represents the correspondence on the Cartesian product of the two shapes. First methods of this type were formulated using graph matching \cite{zeng2010dense}. 
Windheuser et al.\ optimize in a product space \cite{windheuser2011geometrically}, preserving important differential geometric properties. A similar approach was applied in \cite{lahner2016efficient} for 2D-to-3D matching. 
In \cite{vestner2017product}, correspondence is formulated as kernel density estimation on the product manifold, interpreted as an alternating diffusion-sharpening process in \cite{vestnerefficient}. 
A product between more than two shapes is considered in~\cite{multiway}, but the resulting optimization problem is restricted to yield only sparse correspondences.

Soft maps~\cite{solomon2012soft} represent correspondence between shapes as a distribution on the product manifold with prescribed marginals reflecting area preservation.  Nonconvex objectives can be used to incorporate metric information into optimization for soft maps~\cite{memoli2011gromov,solomon2016entropic}, while other objectives on soft maps can be understood as probabilistic relaxations of classical distortion measures from differential geometry~\cite{solomon2013dirichlet,mandad2017variance}. These methods suffer from high complexity, usually quadratic in the number of shape vertices.  

Functional maps~\cite{ovsjanikov2016computing} abandon pointwise correspondence, instead modeling correspondences as linear operators between spaces of functions. An approximation of such operators in a pair of truncated orthogonal bases 
dramatically reduces the problem complexity.  
%
One of the key innovations of this framework is allowing to bring 
a new set of algebraic methods into the domain of shape correspondence. 
Several follow-up works tried to improve the framework by employing sparsity-based priors~\cite{pokrass2013sparse}, manifold optimization~\cite{kovnatsky2013coupled,kovnatsky2016madmm}, non-orthogonal~\cite{kovnatsky2015functional} or localized \cite{choukroun2016elliptic,lmh} bases, 
coupled optimization over the forward and inverse  maps~\cite{eynard2016coupled,ezuz2017deblurring,huang2017adjoint}, combination of functional maps with metric-based approaches \cite{aflalo2016spectral,shamai2016geodesic}\rev{, and kernelization~\cite{wang2018kernel}.} Recent works of~\cite{commutativity,nognengimproved} considered functional algebra (function point-wise multiplications together with addition).   
Generalizations addressing the settings of multiple \rev{shapes}~\cite{huang2014functional,kovnatsky2016madmm}, partial \rev{correspondence}~\cite{rodola2017partial,litany2016non}, \rev{and} cluttered correspondence \cite{cosmo2016matching} have been proposed as well. 
Most recently, functional maps have also been used in conjunction with intrinsic deep learning methods~\cite{litany2017deep}. 
For a comprehensive survey of functional maps and related techniques, we refer the reader to~\cite{ovsjanikov2016computing}.

\paragraph*{Motivation and contribution.} In this paper, we advocate posing correspondence---and understanding relationships between the existing representations above---in terms of functions on the product manifold of the source and target.  
A motivating observation is that functional maps approximate a distribution representing the correspondence in the product space as a linear combination of \emph{separable} tensor-product basis functions. This distribution, however, is supported on a manifold with a dimension \emph{lower} than that of the product space:  For a pair of two dimensional shapes, the distribution is supported on a two-dimensional manifold embedded in a four-dimensional space. Consequently, most of the support of the basis functions is wasted on ``empty'' regions of the product space. Localized bases on the individual domains improve this situation, but still most of their support is wasted. 

We show how point-to-point maps, functional maps, and soft maps all can be understood as (signed) measures on the product and how these representations might be converted to one another.  More importantly, this viewpoint suggests new techniques to \rev{{\em represent} and {\em approximate}} mappings directly on the product, e.g.\ by building a basis from eigenfunctions of the product Laplace--Beltrami operator potentially after filtering undesirable matches.

Our theoretical contributions have practical bearing on the design of correspondence techniques.  After discretizing product manifolds and their Laplace--Beltrami operators, we consider map design and processing problems among two- and three-dimensional shapes.  Reasoning about the product manifold leads to compact, understandable bases for map design that focus resolution in the part of the product most relevant to a correspondence task. One of such means is the construction of \emph{inseparable} bases. To this end, we propose to compute localized harmonics on the product manifold, and discuss a numerical scheme that keeps the complexity of such a computation \rev{feasible and, in particular cases,} comparable to that of the construction of a separable localized basis. 
%
\rev{We finally showcase our framework applied to the task of map refinement.}

\section{Background}\label{sec:background}


\paragraph*{Manifolds.}
We model shapes as Riemannian $d$-manifolds $(\M,g_\M)$ (possibly with boundary $\partial\M$) equipped with  area \rev{elements} $\d x$ induced by the standard metric $g_\M$; we do not restrict our focus to surfaces but rather allow $\M$ and $\N$ to have different intrinsic \rev{dimensions}. We denote by $T_x\M$ the tangent plane at $x\in\M$, modeling the manifold locally as a Euclidean space. Given two scalar functions $f,g : \M\to\mathbb{R}$ belonging to an appropriate functional space $\mathcal{F}(\M)$, we use the standard manifold inner product $\langle f,g \rangle_\M = \int_\M f(x)g(x)\,\d x$. 

In analogy to the Laplace operator in flat spaces, the positive semidefinite Laplace--Beltrami (LB) operator $\Delta_\M$ equips us with the tools to extend Fourier analysis to manifolds. The manifold Laplacian admits an eigen-decomposition $\Delta_\M \phi_i = \lambda_i \phi_i$ for $i \ge 1$, with real eigenvalues $0=\lambda_1\le\lambda_2\le\dots$ and eigenfunctions $\{ \phi_i \}_{i\ge 1}$ forming an orthonormal basis of $L^2(\M)=\{f : \M\to\mathbb{R}~|~\langle f,f\rangle_\M <\infty\}$. Any function $f\in L^2(\M)$ can thus be represented via the Fourier-like series expansion
\begin{equation}
f(x) = \sum_{i\ge 1} \langle f, \phi_i \rangle_\M \phi_i(x)\,.
\end{equation}

\paragraph*{Product manifolds.}
Given two \final{Riemannian} manifolds $(\M,g_\M),(\N,g_\N)$ of dimension $d_\M$ and $d_\N$ \final{with metric tensors $g_\M, g_\N$}, respectively, their product $(\M\times\N,g_\M \oplus g_\N)$ is a manifold of dimension $d_\M+d_\N$, where $g_\M \oplus g_\N = \left(\begin{smallmatrix}g_\M&0\\0 & g_\N\end{smallmatrix}\right)$ is the direct sum of the individual metric tensors \cite{guillemin2010differential}, inducing the area element $\d a = \d x\, \d y$. By this definition of product, to each point $(x,y)\in\M\times\N$ is attached a tangent space derived by the canonical isomorphism $T_{(x,y)}\M \times \N = T_x\M \times T_y\N$ (see \cite[ex.\ 8.7]{tu11}). For tangent vectors $\xi,\eta\in T_x\M$ and $\zeta,\mu\in T_y\N$, the inner product \final{$\langle \cdot, \cdot \rangle$} of $(\xi,\zeta),(\eta,\mu) \in T_{(x,y)}\M \times \N$ is given by
\begin{equation}
\langle (\xi,\zeta) , (\eta,\mu) \rangle_{T_{(x,y)}\M\times \N} = \langle \xi,\eta \rangle_{T_x\M} + \langle \zeta,\mu \rangle_{T_y\N}\,.
\end{equation}
%

Now let $f\in \mathcal{F}(\M)$, $g\in \mathcal{F}(\N)$ for some functional space $\mathcal{F}$, and denote by $f \wedge g$ the outer product of $f$ and $g$ defined by the mapping 
\begin{equation}
f\wedge g : (x,y) \mapsto f(x)g(y)\,.
\end{equation}
The LB operator $\Delta_{\M\times \N}$ obeys the (outer) product rule identity \cite{chavel84}:
\begin{align}
\Delta_{\M\times \N} (f \wedge g) = (\Delta_\M f) \wedge g + f \wedge (\Delta_\N g)\,.
\end{align}
Given 
\rev{eigenvectors $(\phi,\psi)$ with corresponding eigenvalues $(\alpha,\beta)$ satisfying }
$\Delta_\M \phi = \alpha \phi$ and $\Delta_\N \psi = \beta \psi$, application of the product rule yields
\begin{align}\label{eq:prodlap}
\Delta_{\M\times \N} (\phi \wedge \psi) &= (\Delta_\M \phi) \wedge \psi + \phi \wedge (\Delta_\N \psi)\nonumber\\
&=(\alpha + \beta) (\phi \wedge \psi)\,.
\end{align}
This observation leads to a characterization of LB eigenvalues for product manifolds:

\begin{theorem}[\!{\!\cite[Proposition A.II.3]{berger71}}]
Let $\xi$ be an eigenfunction of the product LB operator $\Delta_{\M\times \N}$ with the corresponding eigenvalue $\gamma$. Then, there exist some eigenfunctions $\phi$ of $\Delta_\M$ and $\psi$  of $\Delta_\N$ with the eigenvalues $\alpha$ and $\beta$, respectively, such that $\xi = \phi \wedge \psi$ and $\gamma = \alpha+\beta$. \label{thrm:prodevecs}
\end{theorem}
%

It is also easy to check that the set of eigenfunctions $\{\phi_i \wedge \psi_j\}_{i,j}$ is orthogonal, since:
\begin{align}
\int_{\M\times \N} (\phi_i \wedge \psi_j) &(\phi_k \wedge \psi_\ell)\,\d a 
= \int_{\M\times \N} \phi_i(x)  \psi_j(y) \phi_k(x)   \psi_\ell (y)\,\d a\nonumber\\
&\rev{=} \int_\M \phi_i \phi_k \d x \int_\N \psi_j \psi_\ell \d y\\
&= \delta_{ik} \delta_{j\ell} 
= \left\{ 
		\begin{array}{ll}
			1    & (i = k)\textrm{ and }(j=\ell); \\ 
			0 & \mathrm{otherwise,}
		\end{array}
\right.
\end{align}
where $\delta_{ij}$ is the Kronecker delta.

%
%
%


%
%
%

\paragraph*{Soft maps.}
A soft map $\tilde{\mu}: \M \to \mathrm{Prob}(\N)$ is a function assigning a probability measure over $\N$ to each point in $\M$ \cite{solomon2012soft}. 
Soft maps can be equivalently represented by their densities, i.e., nonnegative scalar functions $\mu:\M \times \N \to [0,1]$ defined on the product manifold $\M\times\N$ satisfying $\tilde{\mu}(x)(B) = \int_{B\subseteq \N} \mu(x,y)\,\d y$ for all $x\in \M$ and all measurable subsets $B\subseteq \N$.

As a particular case, a bijection $\tilde\pi : \M \to \N$ induces a soft map $\tilde{\mu}$ by requiring, for all $x\in \M$, that $\tilde{\mu}(x)(B)=1$ if and only if $\tilde{\pi}(x)\in B \subseteq \N$, i.e., the image $\tilde{\mu}(x)$ is a unit Dirac mass $\delta_{\tilde{\pi}(x)}$ centered at $\tilde{\pi}(x)$.

\paragraph*{Functional maps.}
A functional map $T$ associated to a map $\tilde{\pi} : \M \to \N$ is a linear mapping $T : \mathcal{F} (\N) \to \mathcal{F}(\M)$ defined as  \cite{ovsjanikov12}:
\begin{equation}\label{eq:func}
T(g) = g \circ \tilde{\pi}\,.
\end{equation}

Note how this construction allows to move from identifying a map between manifolds to identifying a linear operator between Hilbert spaces. The functional map $T$ admits a matrix representation wrt orthogonal bases $\{\phi_i\}_{i\ge 1},\{\psi_j\}_{j\ge 1}$ on $\mathcal{F}(\M)$ and $\mathcal{F}(\N)$ respectively, with coefficients $\C = (c_{ij})$ determined as follows:
\begin{equation}\label{eq:c}
T(g) = \sum_{ij \ge 1} \langle \psi_j , g \rangle_\N \underbrace{\langle \phi_i , T(\psi_j) \rangle_\M}_{c_{ij}} \phi_i \,.
\end{equation}
%

\section{Discretization}\label{sec:discretization}
We show how to discretize the main quantities involved in our framework on 1D and 2D manifolds, as well as their products.

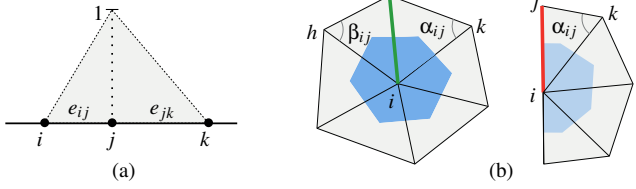
\begin{figure}[bt]
  \centering
  \begin{overpic}
  [trim=0cm 0cm 0cm 0cm,clip,width=0.99\linewidth]{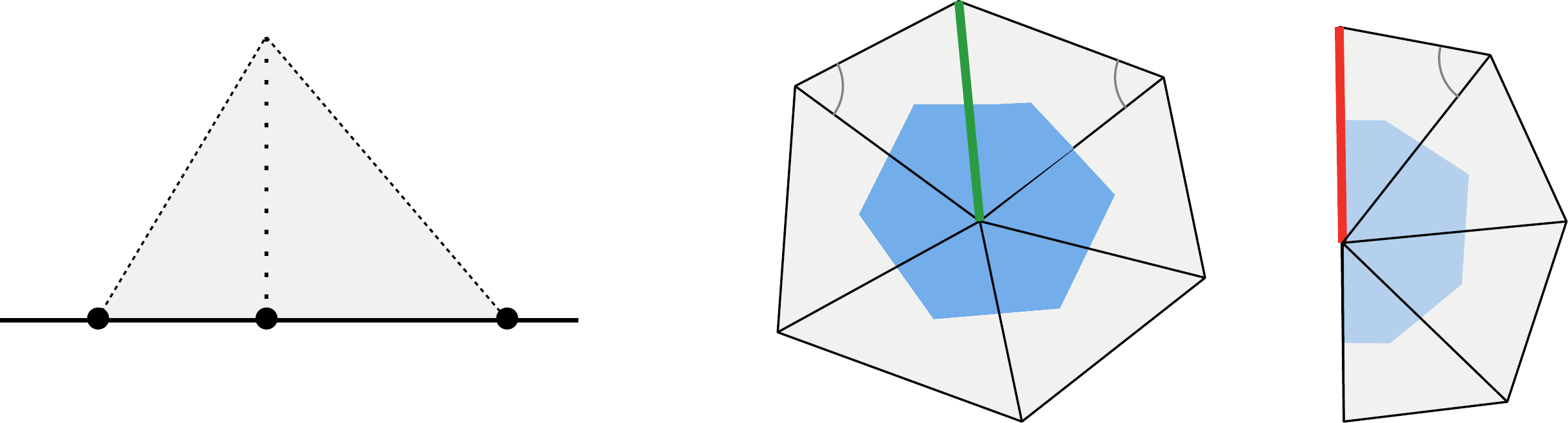}
  \put(5,3){\footnotesize $i$}
  \put(16,3){\footnotesize $j$}
  \put(31,3){\footnotesize $k$}
  \put(10,8){\footnotesize $e_{ij}$}
  \put(23,8){\footnotesize $e_{jk}$}
  \put(14,23){\footnotesize $1$}
  \put(16,24.7){\line(1,0){2}}
  \put(61,9){\footnotesize $i$}
  \put(60,27.5){\footnotesize $j$}
  \put(75,21){\footnotesize $k$}
  \put(48,20){\footnotesize $h$}
  \put(66,21){\footnotesize $\alpha_{ij}$}
  \put(54.5,20){\footnotesize $\beta_{ij}$}
  \put(87,21){\footnotesize $\alpha_{ij}$}
  \put(83.5,10){\footnotesize $i$}
  \put(83.5,25){\footnotesize $j$}
  \put(96,23){\footnotesize $k$}
  \put(17,-2){\footnotesize (a)}
  \put(77,-2){\footnotesize (b)}
  \end{overpic}
  \caption{\label{fig:fem}Discretization of the Laplace-Beltrami operator on a cycle graph {\em (a)} and on a triangle mesh {\em (b)} for interior (green) and boundary edges (red). We also show the hat basis function in {\em (a)}.}
\end{figure}

\paragraph*{1D shapes (curves).}
We model 1D manifolds as closed contours with circular topology (no boundary), discretized as 2-regular cycle graphs $\mathcal{G}=(\mathcal{N},\mathcal{E})$ with $n\ge 3$ nodes $\mathcal{N}$ and as many edges $\mathcal{E}$. The LB operator $\Delta$ is discretized using standard FEM with linear hat functions; in the hat basis, scalar functions on $\mathcal{G}$ are approximated piecewise-linearly on the edges. The Laplacian takes the form of a $n\times n$ sparse matrix $\L=\mathbf{S}^{-1}\mathbf{W}$, where:
\begin{align}
 w_{ij} &=
 \begin{cases} 
       -\frac{1}{\|e_{ij}\|} &  e_{ij} \in \mathcal{E} \\
      -\sum_{i\neq k} w_{ik}  &  i=j\\
      0 & \mathrm{otherwise} 
   \end{cases}
   \\
  s_{ij} &=
 \begin{cases} 
      \frac{1}{6}\|e_{ij}\| &  e_{ij} \in \mathcal{E} \\
	  \frac{1}{3}\sum_{k\in\mathcal{N}(i)} \|e_{ik}\| &  i=j\\
      0 & \mathrm{otherwise} 
   \end{cases}
\end{align}
%
and the notation is according to Figure~\ref{fig:fem}, with $\N(i)$ being the set of the neighbors of node $i$.
%
%
\rev{In} our tests we use non-lumped masses $s_{ij}$\rev{; in applications requiring additional efficiency, lumped mass matrices $\mathrm{diag}(\hat{s}_{ii})$ can be used by setting $\hat{s}_{ii}=\sum_j s_{ij}$.}
%

The product of two boundary-free 1D manifolds $\M,\N$ is a 2D manifold (a surface) $\M\times\N$ with torus topology. For the discretization of the Laplacian on $\M\times\N$, we appeal to the following:
%
\begin{theorem}\label{thm:LB1D}
{\textbf{(Discrete product Laplacian)}}
Let $\M$, $\N$ be 1D manifolds with no boundary, discretized as 2-regular cycle graphs, and let $\S_\M,\W_\M$ and $\S_\N,\W_\N$ be the mass and stiffness matrices for $\Delta_\M$ and $\Delta_\N$ respectively, obtained via FEM with respect to piecewise linear (hat) basis functions. Then,
\begin{align}
\S_{\M\times\N} &= \S_\M \otimes \S_\N\label{eq:sprod}\\
\W_{\M\times\N} &= \W_\M \otimes \S_\N + \S_\M \otimes \W_\N\label{eq:wprod}
\end{align}
are the mass and stiffness matrices for the product manifold Laplacian $\Delta_{\M\times\N}$ with respect to piecewise bilinear basis functions, defined on a quad meshing of the toric surface $\M\times\N$. Here, $\otimes$ denotes the Kronecker product.
\end{theorem}
\begin{proof} See Appendix~\ref{sec:proofs}. \end{proof}

\newtheorem{corollary}{Corollary}
\begin{corollary}\label{thm:L}
The LB operator $\Delta_{\M\times\N}$ is discretized as:
\begin{equation}\label{eq:lb1d}
\L_{\M\times\N} = \L_\M\otimes\mathbf{I}_{\N} + \mathbf{I}_{\M}\otimes\L_\N\,,
\end{equation}
where $\mathbf{I}_{\M},\mathbf{I}_{\N}$ are $n_\M\times n_\M$ and $n_\N\times n_\N$ identity matrices.
\end{corollary}
\begin{proof} See Appendix~\ref{sec:proofs}. \end{proof}

The discretization of $\Delta_{\M\times\N}$ does {\em not} require the explicit construction of a quad mesh embedded in $\mathbb{R}^3$; the toric shapes shown in these pages only serve visualization purposes. Further, the discretization \eqref{eq:lb1d} is consistent with the spectral decomposition identities \eqref{eq:prodlap}; see \cite{fiedler73} and \cite[Proposition 33.6]{hammack11} for additional discussion.

\paragraph*{2D shapes (surfaces).}
We model 2D surfaces as manifold triangle meshes $(\mathcal{V},\mathcal{E},\mathcal{F})$ with $n$ vertices $\mathcal{V}$ connected by edges $\mathcal{E}=\mathcal{E}_\mathrm{i}\cup\mathcal{E}_\mathrm{b}$ (where $\mathcal{E}_\mathrm{i}$ and $\mathcal{E}_\mathrm{b}$ are interior and boundary edges, respectively) and triangle faces $\mathcal{F}$. In analogy to the 1D case, the discretization of the LB operator is obtained using FEM with piecewise linear basis functions on triangle elements~\cite{duffin1959distributed}, taking the form of an $n\times n$ sparse matrix $\L=\mathbf{S}^{-1}\mathbf{W}$, where
\begin{align}
w_{ij} &=
 \begin{cases} 
       (\cot \alpha _{ij} + \cot \beta _{ij})/2 &   ij \in \mathcal{E}_\mathrm{i}  \\
			(\cot \alpha _{ij})/2 &   ij \in \mathcal{E}_\mathrm{b}\\		
			-\sum_{k\neq i} w_{ik}    & i = j\\
			0 & \mathrm{otherwise, and}
   \end{cases}
   \\
s_{ij} &=
 \begin{cases} 
       (A(T_{hij})+A(T_{ijk}))/{12} &   ij \in \mathcal{E}_\mathrm{i}\\
			A(T_{ijk})/{12} &   ij \in \mathcal{E}_\mathrm{b}\\		
			\frac{1}{6}\sum_{k\in \N(i)} A(T_k)     & i = j \\
			0 & \mathrm{otherwise.} 
   \end{cases}
\end{align}
Here, $A(T)$ denotes the area of triangle $T$ and $\N(i)$ is the set of the neighbors of vertex $i$; see Figure~\ref{fig:fem} for notation.

Given two 2D manifolds $\M$ and $\N$, their product is a 4D manifold $\M\times\N$. The LB operator on $\M\times\N$ is discretized similarly to the lower-dimensional case:

\begin{corollary}\label{thm:LB2D}
Let $\M$, $\N$ be surfaces discretized as triangle meshes, and let $\S_\M,\W_\M$ and $\S_\N,\W_\N$ be the mass and stiffness matrices for $\Delta_\M$ and $\Delta_\N$. Then, equations \eqref{eq:sprod}-\eqref{eq:lb1d} provide a valid discretization of the LB operator $\Delta_{\M\times\N}$. This discretization is equivalent to the application of FEM on a 3-3 duoprism tessellation of the 4D product manifold $\M\times\N$ using multilinear basis functions.

\end{corollary}
\begin{proof} See Appendix~\ref{sec:proofs}.\end{proof}

\begin{figure}[bt]
  \centering
    \begin{overpic}
  [trim=0cm 0cm 0cm 0cm,clip,width=0.99\linewidth]{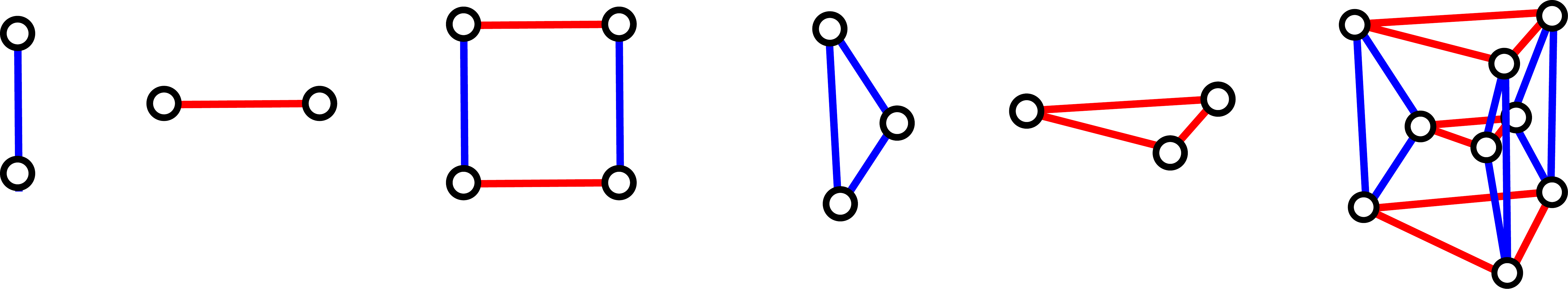}
  \put(4,10.5){$\times$}
  \put(24,10.5){$=$}
  \put(60,10.5){$\times$}
  \put(81,10.5){$=$}
  \put(47,5){\line(0,1){12}}
  \put(18,0){\footnotesize (a)}
  \put(73,0){\footnotesize (b)}
  \end{overpic}
  \caption{\label{fig:products}The Cartesian product of two edge elements is a quad {\em (a)}, while taking the product of two triangles yields a 4D geometric structure called a 3-3 (or triangular) duoprism \final{\cite{polytopes} visualized with a Schlegel diagram \cite{schlegeldiagram}} {\em (b)}. Note that all these objects are polytopes (i.e. they have faces), not simple graphs.}
\end{figure}
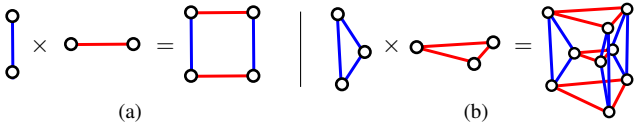

We emphasize that\rev{, as a consequence of the Corollary,} the computation of the product Laplacian \rev{$\Delta_{\M\times\N}$} does not require constructing a high-dimensional embedding for $\M\times\N$, \rev{thus} avoiding cumbersome manipulation of duoprismic product elements (see Figure~\ref{fig:products} for an illustration\rev{ of these elements}).

Finally, scalar functions on a manifold $\M$ are represented by $n$-dimensional vectors $\mathbf{f} = (f(x_1), \hdots, f(x_n))^\top$, where $x_1,\dots,x_n$ denote graph nodes and mesh vertices in the 1D and 2D case respectively. 
Inner products $\langle f,g \rangle_\M$ are discretized as $\mathbf{f}^\top \S \mathbf{g}$, where $\S$ is the mass matrix. On product manifolds, scalar functions are represented as $n_\M\times n_\N$ matrices $\mathbf{F}$, usually deriving from an outer product $f\wedge g$ discretized as $\mathbf{f}\mathbf{g}^\top$; inner products are computed as $\mathrm{vec}(\mathbf{F})^\top \S ~\mathrm{vec}(\mathbf{G})$ \final{($\mathrm{vec}(\mathbf{F})$ stacks the columns of $\mathbf{F}$ into a vector).}

\section{Map representation on the product manifold}

\paragraph*{Soft functional maps.}
It will be instrumental for our purposes to introduce a  ``soft'' generalization of functional maps. For soft maps $\tilde{\mu}: \M \to \mathrm{Prob}(\N)$ with associated density $\mu \in L^1(\M \times \N)$, we define a {\em soft functional map} \rev{$T_\mu : \mathcal{F}(\N) \to \mathcal{F}(\M)$} as the expectation
\begin{align}\label{eq:soft}
T_{\mu}(g)(x) =  \int_\N g(y) \mu(x,y)\,\d y\,.
\end{align}

It is easy to check that  $T_{\mu}$ is linear in $g$, hence admitting a matrix representation with coefficients defined as in \eqref{eq:c}\rev{; in particular, in the standard basis one obtains a stochastic matrix with each row summing to $1$.} If the density $\mu$ encodes a non-soft map (i.e., whenever $\mu(x,\cdot)$ is concentrated at one point), the definition \eqref{eq:soft}  boils down to the original definition \eqref{eq:func}, $T(g)(x) = \int_\N g(y)\,\delta_{\tilde{\pi}(x)}(y)\,\d y = (g \circ \tilde{\pi})(x)$, where the last equivalence stems from the sampling property of Dirac deltas. 

We begin our discussion by deriving a connection between functional map matrices and expanding soft map measures in the Laplace--Beltrami basis:

\begin{theorem}[Equivalence]\label{thm:equivalence}
Let $T_{\mu} : \mathcal{F} (\N) \to \mathcal{F}(\M)$ be a soft functional map \eqref{eq:soft} with underlying density $\mu \in L^1(\M\times\N)$. Further, let $c_{ij} = \langle \phi_i , T_{\mu}(\psi_j) \rangle_\M$ be the matrix coefficients of $T_{\mu}$ in the orthogonal bases $\{\phi_i\}_{i\ge 1},\{\psi_j\}_{j\ge 1}$, and let $p_{ij} = \langle \phi_i \wedge \psi_j, \mu  \rangle_{\M \times \N}$ be the expansion coefficients of $\mu$ in the product basis $\{\phi_i \wedge \psi_j\}_{i,j}$, such that $\mu = \sum_{ij} (\phi_i \wedge \psi_j) p_{ij}$.
Then, $c_{ij} = p_{ij}$ for all $i,j$.
\end{theorem}
%
\begin{proof}
%
%
The functional map matrix coefficients are computed as:
\begin{align}
c_{ij} = \langle \phi_i , T_{\mu}(\psi_j) \rangle_\M &= \int_\M \phi_i (x) T_{\mu}(\psi_j)(x)\,\d x \\
&=\int_\M \phi_i(x) \int_\N \psi_j(y) \mu(x,y)\,\d y\,\d x\\
&=\int_{\M\times \N}\phi_i(x) \psi_j(y) \mu(x,y)\,\d a\,,\label{eq:cij}
\end{align}
%
while the expansion coefficients of $\mu$ are given by
\begin{align}
p_{ij} = \langle \phi_i \wedge \psi_j, \mu  \rangle_{\M \times \N}
= \int_{\M\times \N}\phi_i(x) \psi_j(y)\mu(x,y)\,\d a\,.\label{eq:pij}
\end{align}
Comparing equations \eqref{eq:cij} and \eqref{eq:pij}, we see that $c_{ij} = p_{ij}$ for any choice of $i,j \geq 1$.
\end{proof}

Note that Theorem~\ref{thm:equivalence} applies to any choice of orthogonal bases $\{\phi_i\}_{i\ge 1}\in\mathcal{F}(\M),\{\psi_j\}_{j\ge 1}\in\mathcal{F}(\N)$.


\noindent\textbf{Spectral representation.}
Consider the order-$k$, band-limited approximation of $\mu$:
\begin{equation}\label{eq:ok}
\mu \approx \sum_{\ell=1}^k \xi_\ell p_\ell\,,
\end{equation}
where each $\xi_\ell$ is an eigenfunction of $\Delta_{\M\times\N}$ which uniquely identifies, via \eqref{eq:prodlap}, a pair of eigenfunctions $\phi_i, \psi_j$ on $\M$ and $\N$ respectively. According to Theorem~\ref{thm:equivalence}, the expansion coefficients $p_\ell$ are exactly those appearing in the functional map matrix $\C$, when this is expressed in the Laplacian eigenbases of $\M$ and $\N$ as originally proposed by Ovsjanikov et al. \cite{ovsjanikov12}. There is, however, a crucial difference in the way the two sets of coefficients are stored. We come to the following observation:

\begin{figure}[bt]
  \centering
  \begin{overpic}
  [trim=0cm 0cm 0cm 0cm,clip,width=0.16\linewidth]{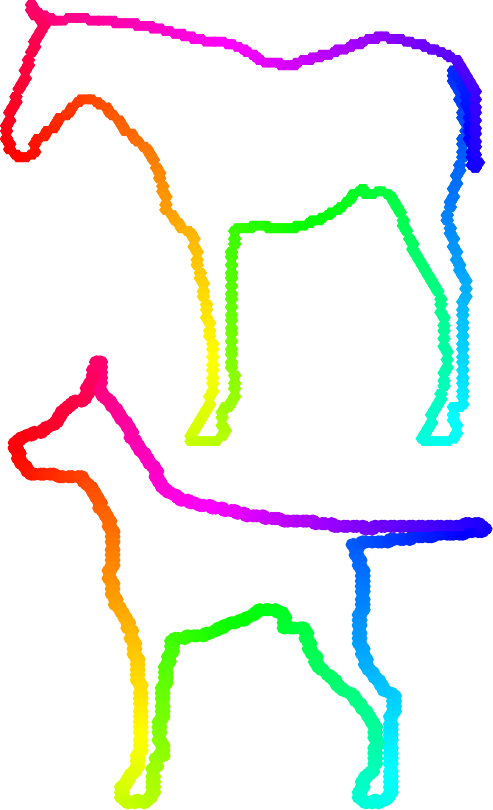}
  \end{overpic}
  \hspace{0.11cm}
  \begin{overpic}
  [trim=0cm 0cm 0cm 0cm,clip,width=0.79\linewidth]{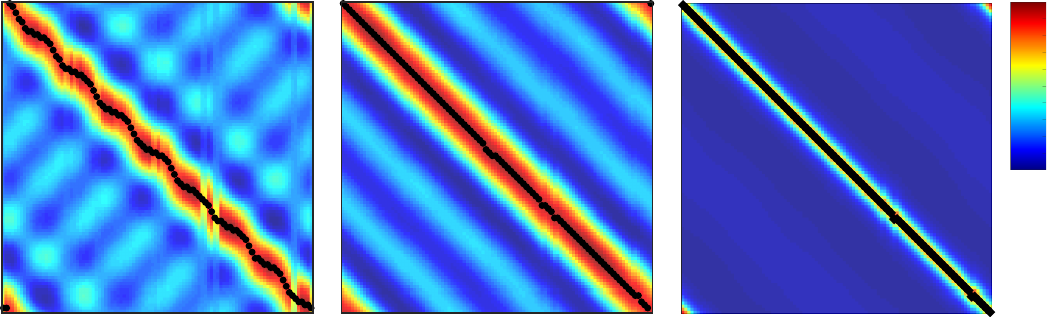}
  \put(12,-5){\footnotesize (a)}
  \put(45,-5){\footnotesize (b)}
  \put(78,-5){\footnotesize (c)}
  \put(100.5,13.8){\tiny 0}
  \put(100.5,27.5){\tiny 1}
  \end{overpic}
  \vspace{0.2cm}
  \caption{\label{fig:approx}The ground truth map (here the identity) between two shapes \rev{approximated according to} {\em (a)} the \rev{standard} functional map representation, {\em (b)} the (separable) LB eigenfunctions of the product manifold, \rev{ordered according to the product eigenvalues,} and {\em (c)} the (inseparable) localized harmonics on the product manifold. \rev{All three cases use the same amount of coefficients.} The black curve in each matrix represents the maximum value for each row. \rev{In this example the product manifold is a flat torus, represented in the parametric domain in {\em (a), (b), (c)}.}}
\end{figure}
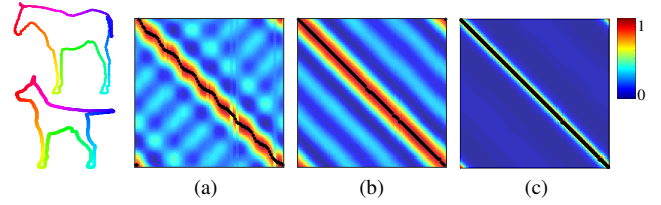

\noindent\textbf{Truncation.}
The product eigenfunctions $\xi_{\ell}$ appearing in the summation \eqref{eq:ok} are associated to the product eigenvalues $\alpha_i+\beta_j$, which are ordered non-decreasingly.
%
%
In contrast, in \cite{ovsjanikov12} it was proposed to truncate the two summations in \eqref{eq:c} to $i=1,\dots,k_\M$ and $j=1,\dots,k_\N$, where indices $i$ and $j$ follow the non-decreasing order of the eigenvalue sequences $\alpha_i$ and $\beta_j$ {\em separately}.

We see that, due to the {\em different ordering}, the eigenfunctions $\phi_i,\psi_j$ involved in the approximation \eqref{eq:ok} of $\mu$ are not necessarily all those involved in the construction of $\C$ \eqref{eq:c}, assuming $k=k_\M k_\N$. In the former case we operate with a reduced basis directly on $\M\times\N$, while in the latter case we consider two reduced bases on $\M$ and $\N$ {\em independently}. This has direct implications on the quality of the approximated maps, as illustrated in Figure~\ref{fig:approx}.

\noindent\textbf{Relation to finite sections.}
The functional map representation was originally introduced in \cite{ovsjanikov12} as a convenient language for solving map inference problems of the type \cite{ovsjanikov2016computing}:
\begin{equation}\label{eq:cab}
\C\A=\B\,,
\end{equation}
where matrices $\B = (\langle \phi_i, f_j \rangle_\M),\A = (\langle \psi_i, g_j \rangle_\N)$ contain Fourier coefficients of a given set of corresponding ``probe'' functions $f_j,g_j,j=1,\dots,q$ on $\M$ and $\N$ respectively (typically, descriptors are used). In the problem above, one is asked to estimate the functional map $\C$.

By truncating the matrix $\C$ to the left upper $k_\M \times k_\N$ submatrix (as in \cite{ovsjanikov12}), one obtains a finite-dimensional approximation of the infinite linear system \eqref{eq:cab}. This procedure, known as the {\em finite section method} \cite{grochenig2010convergence}, does not always guarantee convergence, and a series of remedies using rectangular sections ($k_\M \neq k_\N$) have been proposed in the literature (see \cite{glashoff17} for a discussion pertaining to functional maps).

Recall that, according to Theorem~\ref{thm:equivalence}, the matrix elements $c_{ij}$ correspond to the expansion coefficients $p_{ij}$ appearing in \eqref{eq:ok}. Thus, \rev{due to the different ordering of the $p_{ij}$'s,} the approximation carried out in~\eqref{eq:ok} can be regarded as an ``irregular'' finite section (see Figure~\ref{fig:section}\rev{, right}); in contrast with purely {\em algebraic} approaches considering general systems of linear equations \rev{such as \eqref{eq:cab}}, our approach carries now a {\em geometric} meaning in that \rev{the shape of the section is determined by} the geometry of the product manifold.

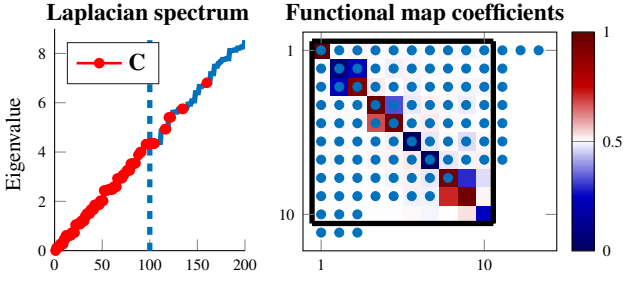
\begin{figure}[bt]
  \centering
%
%
\definecolor{mycolor1}{rgb}{0.00000,0.44700,0.74100}%
\pgfplotsset{scaled y ticks=false}
\begin{tikzpicture}

\begin{axis}[%
width=0.30\columnwidth,
height=0.35\columnwidth,
scale only axis,
xmin=0,
xmax=200,
ymin=-0,
ymax=0.0009,
yticklabels={0,0,2,4,6,8},
ylabel={\footnotesize Eigenvalue},
ylabel style={at={(0.28,0.48)}},
every x tick label/.append style={font=\color{black}, font=\tiny},
every y tick label/.append style={font=\color{black}, font=\tiny},
axis background/.style={fill=white},
title style={font=\bfseries, at={(0.49,0.91)}},
title={Laplacian spectrum},
axis x line*=bottom,
axis y line*=left,
legend style={legend cell align=left, align=left, draw=white!15!black, at={(0.55,0.94)}}
]
\addplot [color=mycolor1, line width=2.0pt, forget plot]
  table[row sep=crcr]{%
1	-0\\
2	1.16162882193294e-05\\
3	1.16163948575831e-05\\
4	1.54018580644788e-05\\
5	1.54019636511293e-05\\
6	2.70181462838082e-05\\
9	2.70183585087125e-05\\
10	4.65135969704988e-05\\
11	4.65155829090236e-05\\
12	6.16700909574774e-05\\
17	6.19175465601529e-05\\
18	7.32863791768068e-05\\
21	7.32879838665212e-05\\
22	0.000104837796612856\\
23	0.000104847409829745\\
24	0.000108183687927976\\
27	0.00010818717188954\\
28	0.000120239654677334\\
31	0.000120249373480874\\
32	0.000138994674415471\\
33	0.000138998555115677\\
34	0.000150610962606379\\
37	0.00015061494997326\\
38	0.000166507887541911\\
41	0.000166518998838683\\
42	0.000185508271357548\\
45	0.000185514138024701\\
46	0.000186825950777347\\
47	0.00018686893045583\\
48	0.000202227808841826\\
51	0.00020227089410696\\
52	0.000243832470999905\\
55	0.000243845964945422\\
56	0.00024766625622874\\
57	0.000247726559962302\\
58	0.000248496041734825\\
61	0.000248540519464768\\
62	0.00025928254444807\\
65	0.000259342954819886\\
66	0.000292846150614423\\
67	0.000292906257641334\\
68	0.000294179853199239\\
71	0.000294242142871326\\
72	0.000308248008678902\\
75	0.000308308221292464\\
76	0.000325820625192819\\
79	0.000325867485571507\\
80	0.000352504052841596\\
83	0.000352573969792047\\
84	0.000354516241571901\\
87	0.000354577846621851\\
88	0.000388184616156195\\
89	0.000388247137721009\\
90	0.000399800904347103\\
93	0.000399863532578593\\
94	0.000423351787446791\\
95	0.000423403363328134\\
96	0.000431840825001473\\
99	0.00043190481272859\\
100	0.000434492207006087\\
107	0.000434762720630033\\
108	0.000438753645511269\\
111	0.000438805326979264\\
112	0.000485021878375846\\
115	0.000485074952308651\\
116	0.000493022412740629\\
119	0.000493094547550754\\
120	0.000540512406843163\\
123	0.000540632817603637\\
124	0.000561015849314117\\
125	0.000561245669615573\\
126	0.00056234646183384\\
129	0.000562401918443811\\
130	0.000572632137533446\\
133	0.000572862064473156\\
134	0.000575010566933543\\
137	0.00057511606817684\\
138	0.000578663990978612\\
139	0.000579141390090854\\
140	0.000594065849043091\\
141	0.000594065954629741\\
142	0.000594543248155333\\
143	0.000594543353741983\\
144	0.000607529446284616\\
147	0.000607761252524597\\
148	0.000640334081936089\\
149	0.00064033557998755\\
150	0.000640811481019909\\
151	0.000640812979099792\\
152	0.000665853645926973\\
155	0.000666093079445318\\
156	0.000671018043675531\\
159	0.000671129923290437\\
160	0.000681030766742197\\
163	0.000681153395362344\\
164	0.000717658665394083\\
165	0.000717662546094289\\
166	0.000718136064477903\\
167	0.000718139945206531\\
168	0.000747841800091464\\
171	0.000748114600071403\\
172	0.000759850349652424\\
173	0.000760364367636157\\
174	0.0007670763823171\\
175	0.000767221495749482\\
176	0.000775252207716903\\
177	0.000775252313331976\\
178	0.000775766225700636\\
179	0.000775766331287286\\
180	0.000778692670508008\\
183	0.000778837890607065\\
184	0.000811536403574564\\
187	0.000811650501049144\\
188	0.000813589979259177\\
191	0.000813737078658505\\
192	0.000821520440609902\\
193	0.000821521938661363\\
194	0.000822034458565213\\
195	0.000822035956616674\\
196	0.000826330247207352\\
199	0.000826867950053156\\
200	0.00085386199992854\\
};

\addplot [color=red, line width=1.0pt, draw=none, mark size=1.5pt, mark=*, mark options={solid, fill=red}]
  table[row sep=crcr]{%
1	-0\\
2	1.16162882193294e-05\\
3	1.16163948575831e-05\\
10	4.65135969704988e-05\\
11	4.65155829090236e-05\\
22	0.000104837796612856\\
23	0.000104847409829745\\
46	0.000186825950777347\\
47	0.00018686893045583\\
66	0.000292846150614423\\
4	1.54018580644788e-05\\
6	2.70181462838082e-05\\
8	2.70182529220619e-05\\
14	6.19154550349776e-05\\
16	6.19174409735024e-05\\
28	0.000120239654677334\\
30	0.000120249267894224\\
48	0.000202227808841826\\
50	0.000202270788520309\\
72	0.000308248008678902\\
5	1.54019636511293e-05\\
7	2.70182518704587e-05\\
9	2.70183585087125e-05\\
15	6.19155606216282e-05\\
17	6.19175465601529e-05\\
29	0.000120239760263985\\
31	0.000120249373480874\\
49	0.000202227914428477\\
51	0.00020227089410696\\
73	0.000308248114265552\\
12	6.16700909574774e-05\\
18	7.32863791768068e-05\\
19	7.32864858150606e-05\\
24	0.000108183687927976\\
26	0.000108185673838079\\
38	0.000166507887541911\\
40	0.000166517500787222\\
58	0.000248496041734825\\
60	0.000248539021413308\\
84	0.000354516241571901\\
13	6.16715890089381e-05\\
20	7.32878772282675e-05\\
21	7.32879838665212e-05\\
25	0.000108185185979437\\
27	0.00010818717188954\\
39	0.000166509385621794\\
41	0.000166518998838683\\
59	0.000248497539786285\\
61	0.000248540519464768\\
85	0.000354517739623361\\
32	0.000138994674415471\\
34	0.000150610962606379\\
35	0.000150611069273054\\
42	0.000185508271357548\\
43	0.000185510257296073\\
52	0.000243832470999905\\
54	0.000243842084245216\\
76	0.000325820625192819\\
78	0.000325863604871302\\
96	0.000431840825001473\\
33	0.000138998555115677\\
36	0.000150614843335006\\
37	0.00015061494997326\\
44	0.000185512152086176\\
45	0.000185514138024701\\
53	0.000243836351728532\\
55	0.000243845964945422\\
77	0.000325824505893024\\
79	0.000325867485571507\\
97	0.0004318447057301\\
56	0.00024766625622874\\
62	0.00025928254444807\\
63	0.000259282651086323\\
68	0.000294179853199239\\
69	0.000294181839137764\\
80	0.000352504052841596\\
81	0.000352513666058485\\
100	0.000434492207006087\\
101	0.00043453518668457\\
120	0.000540512406843163\\
57	0.000247726559962302\\
64	0.000259342848181632\\
65	0.000259342954819886\\
70	0.000294240156932801\\
71	0.000294242142871326\\
82	0.000352564356575158\\
83	0.000352573969792047\\
102	0.00043455251073965\\
103	0.000434595490418133\\
122	0.000540572710576726\\
88	0.000388184616156195\\
90	0.000399800904347103\\
91	0.000399801010985357\\
104	0.000434698213098272\\
105	0.000434700199036797\\
116	0.000493022412740629\\
117	0.00049303202598594\\
134	0.000575010566933543\\
135	0.000575053546612025\\
160	0.000681030766742197\\
};
\addlegendentry{$\C$}

\addplot [color=mycolor1, dashed, line width=2.0pt, forget plot]
  table[row sep=crcr]{%
100	0\\
100	0.00085386199992854\\
};
\end{axis}
\end{tikzpicture}%
%
%
\definecolor{mycolor1}{rgb}{0.00000,0.44700,0.74100}%
\begin{tikzpicture}

\begin{axis}[%
width=0.4\columnwidth,
scale only axis,
axis on top,
xmin=0,
xmax=14,
xlabel style={font=\color{white!15!black}},
y dir=reverse,
ymin=0,
ymax=12,
xtick={1,10},
xticklabels={1,10},
ytick={1,10},
yticklabels={1,10},
every x tick label/.append style={font=\color{black}, font=\tiny},
every y tick label/.append style={font=\color{black}, font=\tiny},
axis background/.style={fill=white},
title style={font=\bfseries, at={(0.48,0.92)}},
title={Functional map coefficients},
legend={},
colormap={mymap}{[1pt] rgb(0pt)=(0,0,0.372549); rgb(1pt)=(0,0,0.380392); rgb(2pt)=(0,0,0.384314); rgb(3pt)=(0,0,0.392157); rgb(4pt)=(0,0,0.396078); rgb(5pt)=(0,0,0.403922); rgb(6pt)=(0,0,0.407843); rgb(8pt)=(0,0,0.423529); rgb(9pt)=(0,0,0.427451); rgb(10pt)=(0,0,0.435294); rgb(11pt)=(0,0,0.439216); rgb(12pt)=(0,0,0.447059); rgb(13pt)=(0,0,0.45098); rgb(15pt)=(0,0,0.466667); rgb(16pt)=(0,0,0.470588); rgb(17pt)=(0,0,0.478431); rgb(18pt)=(0,0,0.482353); rgb(19pt)=(0,0,0.490196); rgb(20pt)=(0,0,0.494118); rgb(21pt)=(0,0,0.501961); rgb(22pt)=(0,0,0.505882); rgb(24pt)=(0,0,0.521569); rgb(25pt)=(0,0,0.52549); rgb(26pt)=(0,0,0.533333); rgb(27pt)=(0,0,0.537255); rgb(28pt)=(0,0,0.545098); rgb(29pt)=(0,0,0.54902); rgb(31pt)=(0,0,0.564706); rgb(32pt)=(0,0,0.568627); rgb(33pt)=(0,0,0.576471); rgb(34pt)=(0,0,0.580392); rgb(35pt)=(0,0,0.588235); rgb(36pt)=(0,0,0.592157); rgb(37pt)=(0,0,0.6); rgb(38pt)=(0,0,0.603922); rgb(40pt)=(0,0,0.619608); rgb(41pt)=(0,0,0.623529); rgb(42pt)=(0,0,0.631373); rgb(43pt)=(0,0,0.635294); rgb(44pt)=(0,0,0.643137); rgb(45pt)=(0,0,0.647059); rgb(47pt)=(0,0,0.662745); rgb(48pt)=(0,0,0.666667); rgb(49pt)=(0,0,0.67451); rgb(50pt)=(0,0,0.678431); rgb(51pt)=(0,0,0.686275); rgb(52pt)=(0,0,0.690196); rgb(53pt)=(0,0,0.698039); rgb(54pt)=(0,0,0.701961); rgb(56pt)=(0,0,0.717647); rgb(57pt)=(0,0,0.721569); rgb(58pt)=(0,0,0.729412); rgb(59pt)=(0,0,0.733333); rgb(60pt)=(0,0,0.741176); rgb(61pt)=(0,0,0.745098); rgb(63pt)=(0,0,0.760784); rgb(64pt)=(0.00392157,0.00392157,0.764706); rgb(72pt)=(0.129412,0.129412,0.796078); rgb(73pt)=(0.145098,0.145098,0.796078); rgb(90pt)=(0.411765,0.411765,0.862745); rgb(91pt)=(0.427451,0.427451,0.862745); rgb(109pt)=(0.709804,0.709804,0.933333); rgb(110pt)=(0.72549,0.72549,0.933333); rgb(127pt)=(0.992157,0.992157,1); rgb(128pt)=(1,0.992157,0.992157); rgb(160pt)=(0.87451,0.490196,0.490196); rgb(161pt)=(0.87451,0.47451,0.47451); rgb(191pt)=(0.756863,0.00392157,0.00392157); rgb(192pt)=(0.752941,0,0); rgb(193pt)=(0.745098,0,0); rgb(194pt)=(0.741176,0,0); rgb(195pt)=(0.733333,0,0); rgb(196pt)=(0.729412,0,0); rgb(197pt)=(0.721569,0,0); rgb(198pt)=(0.717647,0,0); rgb(199pt)=(0.709804,0,0); rgb(200pt)=(0.705882,0,0); rgb(201pt)=(0.698039,0,0); rgb(202pt)=(0.694118,0,0); rgb(203pt)=(0.686275,0,0); rgb(204pt)=(0.682353,0,0); rgb(206pt)=(0.666667,0,0); rgb(207pt)=(0.662745,0,0); rgb(208pt)=(0.654902,0,0); rgb(209pt)=(0.65098,0,0); rgb(210pt)=(0.643137,0,0); rgb(211pt)=(0.639216,0,0); rgb(212pt)=(0.631373,0,0); rgb(213pt)=(0.627451,0,0); rgb(214pt)=(0.619608,0,0); rgb(215pt)=(0.615686,0,0); rgb(216pt)=(0.607843,0,0); rgb(217pt)=(0.603922,0,0); rgb(219pt)=(0.588235,0,0); rgb(220pt)=(0.584314,0,0); rgb(221pt)=(0.576471,0,0); rgb(222pt)=(0.572549,0,0); rgb(223pt)=(0.564706,0,0); rgb(224pt)=(0.560784,0,0); rgb(225pt)=(0.552941,0,0); rgb(226pt)=(0.54902,0,0); rgb(227pt)=(0.541176,0,0); rgb(228pt)=(0.537255,0,0); rgb(229pt)=(0.529412,0,0); rgb(230pt)=(0.52549,0,0); rgb(232pt)=(0.509804,0,0); rgb(233pt)=(0.505882,0,0); rgb(234pt)=(0.498039,0,0); rgb(235pt)=(0.494118,0,0); rgb(236pt)=(0.486275,0,0); rgb(237pt)=(0.482353,0,0); rgb(238pt)=(0.47451,0,0); rgb(239pt)=(0.470588,0,0); rgb(240pt)=(0.462745,0,0); rgb(241pt)=(0.458824,0,0); rgb(242pt)=(0.45098,0,0); rgb(243pt)=(0.447059,0,0); rgb(245pt)=(0.431373,0,0); rgb(246pt)=(0.427451,0,0); rgb(247pt)=(0.419608,0,0); rgb(248pt)=(0.415686,0,0); rgb(249pt)=(0.407843,0,0); rgb(250pt)=(0.403922,0,0); rgb(251pt)=(0.396078,0,0); rgb(252pt)=(0.392157,0,0); rgb(253pt)=(0.384314,0,0); rgb(254pt)=(0.380392,0,0); rgb(255pt)=(0.372549,0,0)},
colorbar,
colorbar style={
		at={(1.06,1.0)},
		width=0.035\columnwidth,
        ytick={0,0.5,1},
        yticklabel style={
            font=\tiny,
        }
    }
]
\addplot [forget plot] graphics [xmin=0.5, xmax=10.5, ymin=0.5, ymax=10.5] {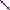};
\addplot [color=mycolor1, draw=none, mark size=1.7pt, mark=*, mark options={solid, mycolor1}]
  table[row sep=crcr]{%
1	1\\
1	2\\
1	3\\
1	4\\
1	5\\
1	6\\
1	7\\
1	8\\
1	9\\
1	10\\
1	11\\
2	1\\
2	2\\
2	3\\
2	4\\
2	5\\
2	6\\
2	7\\
2	8\\
2	9\\
2	10\\
2	11\\
3	1\\
3	2\\
3	3\\
3	4\\
3	5\\
3	6\\
3	7\\
3	8\\
3	9\\
3	10\\
3	11\\
4	1\\
4	2\\
4	3\\
4	4\\
4	5\\
4	6\\
4	7\\
4	8\\
4	9\\
5	1\\
5	2\\
5	3\\
5	4\\
5	5\\
5	6\\
5	7\\
5	8\\
5	9\\
6	1\\
6	2\\
6	3\\
6	4\\
6	5\\
6	6\\
6	7\\
6	8\\
6	9\\
7	1\\
7	2\\
7	3\\
7	4\\
7	5\\
7	6\\
7	7\\
7	8\\
7	9\\
8	1\\
8	2\\
8	3\\
8	4\\
8	5\\
8	6\\
8	7\\
8	8\\
9	1\\
9	2\\
9	3\\
9	4\\
9	5\\
9	6\\
9	7\\
10	1\\
10	2\\
10	3\\
10	4\\
10	5\\
10	6\\
10	7\\
11	1\\
11	2\\
11	3\\
11	4\\
11	5\\
11	6\\
11	7\\
12	1\\
13	1\\
};

\addplot [color=black, line width=2.0pt]
  table[row sep=crcr]{%
0.5	10.5\\
10.5	10.5\\
};

\addplot [color=black, line width=2.0pt]
  table[row sep=crcr]{%
0.5	0.5\\
10.5	0.5\\
};

\addplot [color=black, line width=2.0pt]
  table[row sep=crcr]{%
0.5	0.5\\
0.5	10.5\\
};

\addplot [color=black, line width=2.0pt]
  table[row sep=crcr]{%
10.5	0.5\\
10.5	10.5\\
};

\end{axis}
\end{tikzpicture}%
  \caption{\label{fig:section}{\em Left}: The $k=100$ frequencies involved in the construction of a $10\times 10$ functional map matrix $\C$ correspond to an irregular sampling of the Laplacian spectrum of the product manifold. {\em Right}: In turn, only some of the coefficients $c_{ij}$ of matrix $\C$ appear among the {\em first} $k$ expansion coefficients $p_{ij}$ of the map in the product eigenbasis. Here $\C$ is framed in black, while the blue dots identify the first $k$ coefficients $p_{ij}$.}
\end{figure}

\section{Spectral map processing}
In this paper, we consider curves and surfaces as our shapes. Despite their different intrinsic dimensions, our framework applies to both without specific adjustment.

\paragraph*{Localized spectral encoding.}
%


Theorem~\ref{thm:equivalence} establishes the equivalence between the soft functional map $T_{\mu}$ representation coefficients $c_{ij}$ in the bases $\{ \phi_i\}_{i\geq 1} \subseteq \mathcal{F}(\M)$ and $\{ \psi_j\}_{j\geq 1} \subseteq \mathcal{F}(\N)$ and the coefficients $p_{\ell}$ of the underlying density $\mu$ Fourier series~(\ref{eq:ok}) in the eigenbasis $\{ \xi_{\ell}\}_{\ell\geq 1} \subseteq \mathcal{F}(\M\times \N)$ of the product manifold Laplacian $\Delta_{\M \times \N}$. This equivalence directly stems from $\xi_{\ell}$'s having the separable form $\phi_i \wedge \psi_j$, by virtue of Theorem~\ref{thrm:prodevecs}.
It may be advantageous, however, to consider different orthonormal bases on $\M \times \N$ that are not necessarily separable.
In particular, we observe that $\mu$ tends to be localized on the product manifold $\M \times \N$ (see Figure~\ref{fig:approx}), and thus the standard outer product basis is extremely wasteful as it is supported on the entire $\M \times \N$.

%
A better alternative is the use of {\em localized manifold harmonics} \cite{choukroun2016elliptic,lmh}. Assume that we are given a rough indication of the support of $\mu$ \rev{(for example, coming from a shape matching algorithm)} in the form of a step potential function
\begin{equation}
V(x,y) = \left\{
\begin{array}{ll}
			\nu    & \mu(x,y) \approx 0; \\
			0 & \mathrm{otherwise.}
		\end{array}
\right.
\end{equation}
where $\nu\geq 1$. Then, the variational problem
\begin{eqnarray}
\min_{ \xi_1, \hdots, \xi_k } && \sum_{\ell=1}^k \int_{\M\times \N} \left(\| \nabla_{\M\times \N} \xi_{\ell} \|^2_{g_\M \oplus g_\N} + V \xi_{\ell}^2\right)\,\d a \\
\mathrm{s.t.} && \langle \xi_{\ell}, \xi_{\ell'}\rangle_{\M\times \N} = \delta_{\ell,\ell'} \nonumber
\end{eqnarray}
produces a set of orthonormal functions denoted by $\hat{\xi}_1,\hdots, \hat{\xi}_k$ that, for a sufficiently large value of $\nu$, are also localized in the support of $V$. Note that this new basis $\{ \hat{\xi}_{\ell}\}_{\ell=1}^k$ is {\em no \rev{longer} separable}, i.e., the functions $\hat{\xi}$ are not in general expressible as outer products of functions defined on the originating domains. See Figures~\ref{fig:bases} and \ref{fig:proj} for an illustration, and Figures~~\ref{fig:support} and \ref{fig:delta3d} for practical examples.

\begin{figure*}[t]
  \centering
  \begin{overpic}
  [trim=0cm 0cm 0cm 0cm,clip,width=0.9\linewidth]{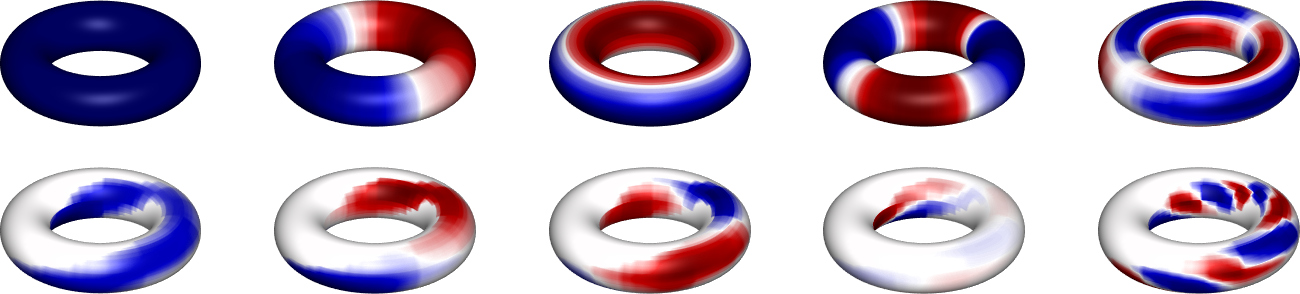}
  \put(2,-2){\footnotesize Eigenfunction $1$}
  \put(28,-2){\footnotesize $3$}
  \put(49,-2){\footnotesize $5$}
  \put(70,-2){\footnotesize $10$}
  \put(92,-2){\footnotesize $20$}
  \end{overpic}
  \vspace{0.16cm}
  \caption{\label{fig:bases}Basis functions on the product manifold (here visualized as a torus embedded in $\mathbb{R}^3$) of two 1D shapes. We plot a few standard LB eigenfunctions (top row) and localized manifold harmonics (bottom row). Here and in the following, we use the present color scheme (blue denotes small values, red large values, white is zero).}
\end{figure*}
\begin{figure}[bt]
  \centering
  \begin{overpic}
  [trim=0cm 0cm 0cm 0cm,clip,width=0.29\linewidth]{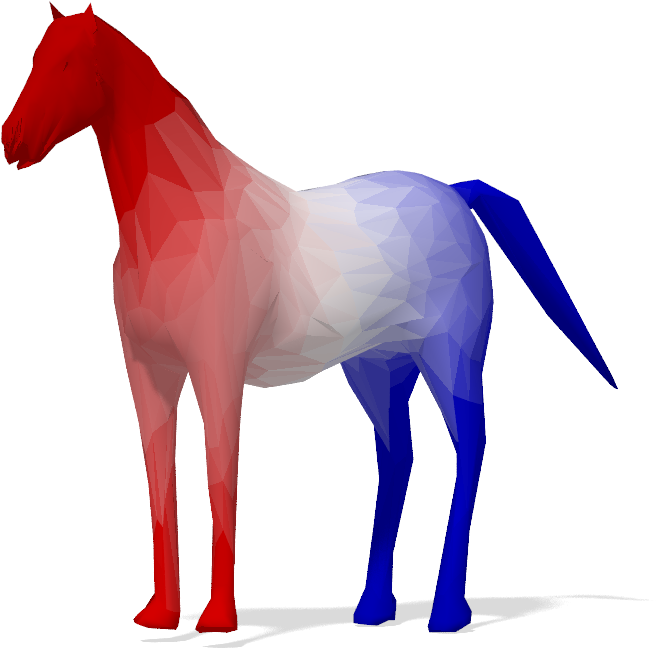}
  \end{overpic}
  \begin{overpic}
  [trim=0cm 0cm 0cm 0cm,clip,width=0.29\linewidth]{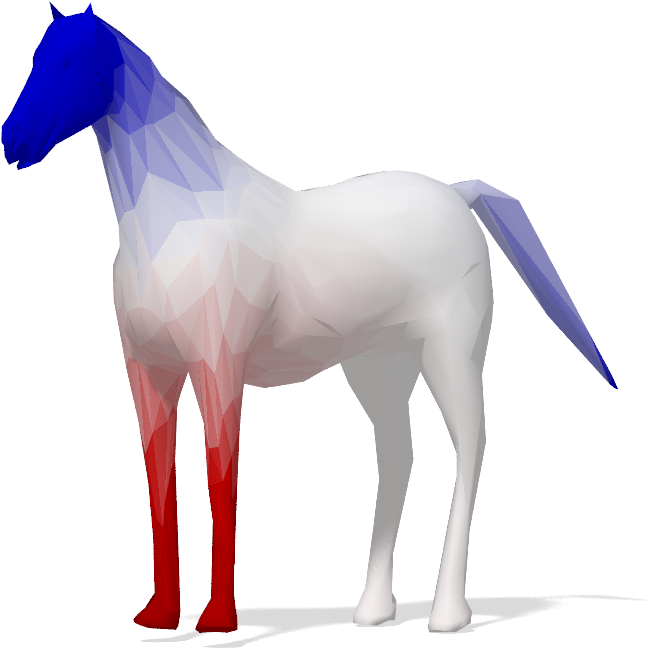}
  \end{overpic}
  \begin{overpic}
  [trim=0cm 0cm 0cm 0cm,clip,width=0.29\linewidth]{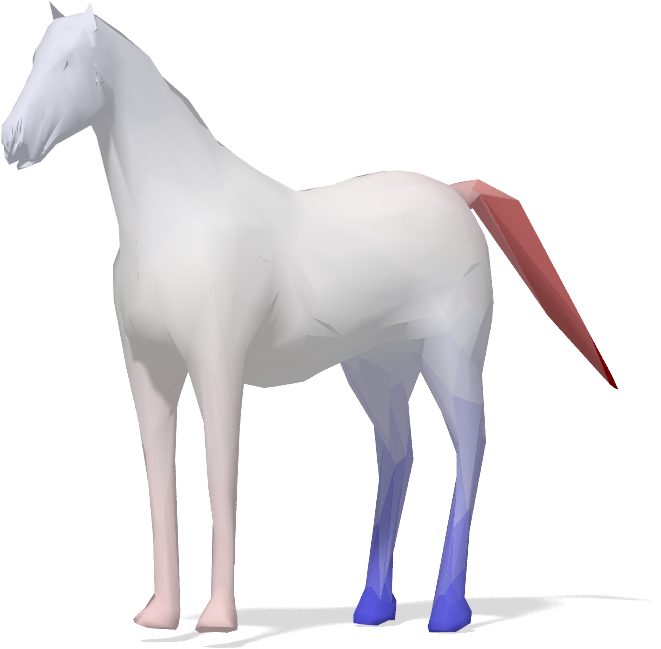}
  \end{overpic}
  \begin{overpic}
  [trim=0cm 0cm 0cm 0cm,clip,width=0.29\linewidth]{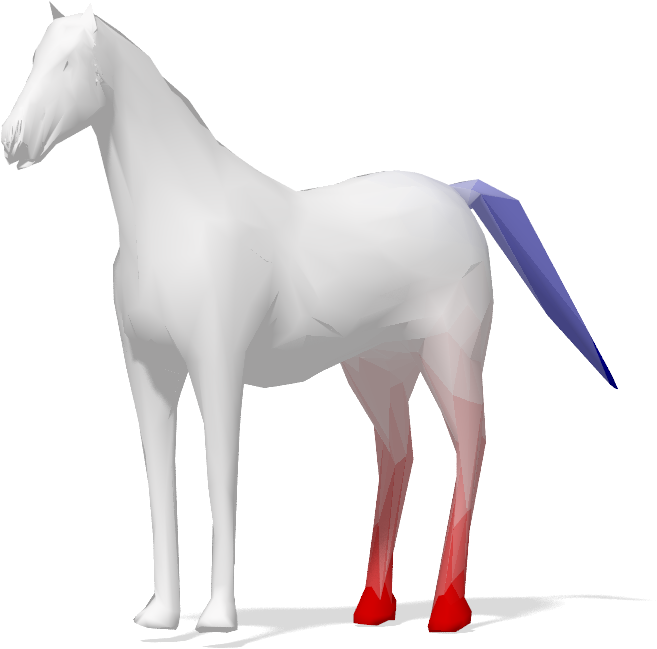}
  \end{overpic}
  \begin{overpic}
  [trim=0cm 0cm 0cm 0cm,clip,width=0.29\linewidth]{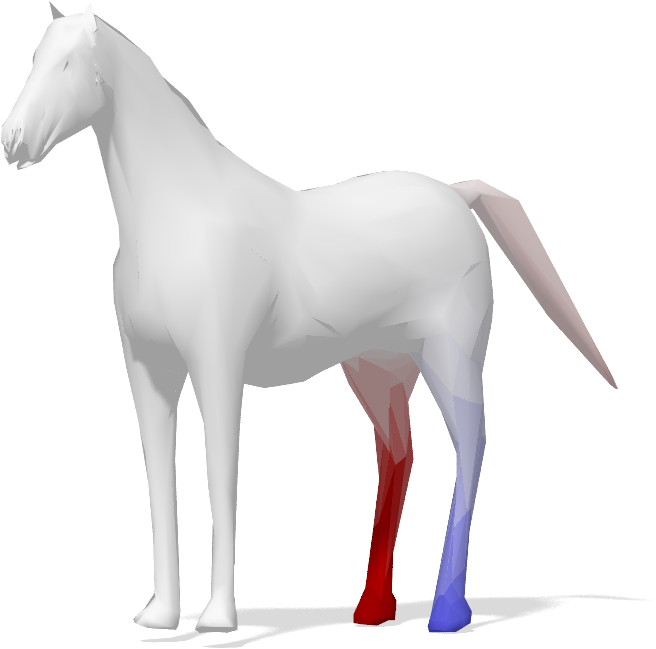}
  \end{overpic}
  \begin{overpic}
  [trim=0cm 0cm 0cm 0cm,clip,width=0.29\linewidth]{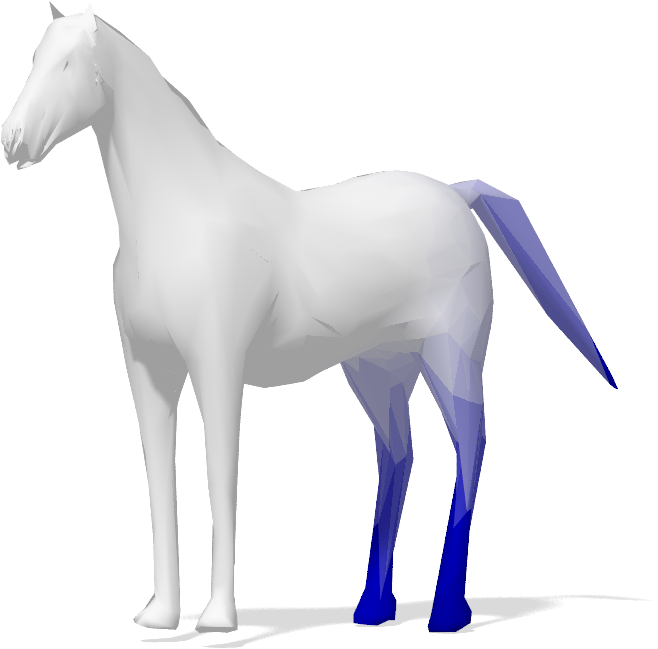}
  \end{overpic}
  \caption{\label{fig:proj}Projecting the basis functions on the product manifold of horse and elephant back onto the factor shapes (here only the horse projection is visualized). {\em Top row}: Projection of three product LB eigenfunctions, which correspond exactly to three standard LB eigenfunctions on the horse shape. {\em Bottom row}: Projection of three localized harmonics; these projections do {\em not} correspond to any LB eigenfunction on the horse. Still, note how they capture the geometric features of the underlying shape.}
\end{figure}
The basis $\{ \hat{\xi}_{\ell}\}_{\ell=1}^k$ turns out to be the eigenbasis of the {\em Hamiltonian operator} \cite{choukroun2016elliptic} $H = \Delta_{\M\times \N} + V$ and can be computed by the eigendecomposition of the product Laplacian matrix with the addition of diagonal potential. 
The size of such problem can be huge (if the shapes are discretized with $n \sim 10^3$ points, the product Laplacian matrix has size $n^2\times n^2 = 10^6 \times 10^6$; see Theorem~\ref{thm:LB1D}), and despite its extreme sparsity, computationally expensive.

As an alternative, we consider a {\em patch} $\mathcal{P} \subset \M \times \N$ of the product manifold \final{with boundary $\partial \mathcal{P}$} corresponding to \rev{$\mu(x,y)>0$}, and define the eigenproblem
\begin{eqnarray}\label{eq:patchbc}
\begin{array}{ll}
\Delta_{\mathcal{P}} \bar{\xi}_{\ell}(x,y) = \gamma_{\ell} \bar{\xi}_{\ell}(x,y) & (x,y) \in \mathrm{int}(\mathcal{P})\\
\bar{\xi}_{\ell}(x,y) = 0 & (x,y) \in  \partial \mathcal{P} 
		\end{array}
\end{eqnarray}
%
of the {\em product patch Laplacian} $\Delta_{\mathcal{P}}$ with \rev{homogeneous} Dirichlet boundary conditions.
%
%
\rev{In practice, this is implemented by constructing the stiffness and mass matrices $\mathbf{W}_{\mathrm{int}(\mathcal{P})}, \mathbf{S}_{\mathrm{int}(\mathcal{P})}$ by selecting the rows and columns of $\mathbf{W}_{\M\times\N}, \mathbf{S}_{\M\times\N}$ that correspond to the vertices in $\mathrm{int}(\mathcal{P})$. A generalized eigenproblem \final{using $\mathbf{W}_{\mathrm{int}(\mathcal{P})}, \mathbf{S}_{\mathrm{int}(\mathcal{P})}$} is solved, yielding eigenfunctions $\bar{\xi}_{\mathrm{int}(\mathcal{P})}$ defined on $\mathrm{int}(\mathcal{P})$; the final eigenfunctions $\bar{\xi}$ on the entire patch $\mathcal{P}$ are obtained by setting $\bar{\xi}(x)=\xi_{\mathrm{int}(\mathcal{P})}$ for $x\in\mathrm{int}(\mathcal{P})$ and $\bar{\xi}(x)=0$ for $x\in\partial\mathcal{P}$.}

If the patch is selected in such a way that its size scales as $\mathcal{O}(n)$ rather than $\mathcal{O}(n^2)$ in the size of the shapes (in practice, this can be achieved by taking a fixed-size band around the initial correspondence), the computation of the localized basis $\{ \bar{\xi}_{\ell} \}_{\ell=1}^k$ has the same complexity as eigendecomposition of the individual Laplacians $\Delta_\M, \Delta_\N$.
\rev{An example application of this construction is described next.}

\rev{Despite the computational gains of working with patches $\mathcal{P}\subset\M\times\N$, computing the eigen-decomposition of the full Hamiltonian $\Delta_{\M\times\N}+V$ may still be useful in certain settings. Note, in particular, that one may define a {\em soft} potential $V(x,y) = 1-\mu(x,y)$ \cite{lmh} directly reflecting the reliability of the underlying map in terms of its density. Further, it is also possible to define a {\em patch Hamiltonian} $\Delta_\mathcal{P} + V|_\mathcal{P}$ with soft potential if desired.}

\begin{figure*}[bt]
  \centering
\begin{minipage}{0.2\linewidth}
\vspace{5mm}
\setlength\figureheight{0.75\linewidth}
\setlength\figurewidth{\linewidth}
%
%
\definecolor{mycolor1}{rgb}{0.00000,0.33333,0.83333}%
\definecolor{mycolor2}{rgb}{0.00000,0.66667,0.66667}%
\definecolor{mycolor3}{rgb}{0.00000,1.00000,0.50000}%
\pgfplotsset{scaled x ticks=false}
\begin{tikzpicture}

\begin{axis}[%
width=\figurewidth,
height=\figureheight,
scale only axis,
xmin=0,
xmax=0.1,
xticklabels={0,0,0.02,0.04,0.06,0.08,0.1},
xlabel style={font=\color{white!15!black}},
xlabel style={at={(0.5,0.07)}},
xlabel={\footnotesize Geodesic error},
ymin=0,
ymax=100,
ylabel style={font=\color{white!15!black}},
ylabel style={at={(0.165,0.48)}},
ylabel={\footnotesize \% Correspondences},
every x tick label/.append style={font=\color{black}, font=\tiny},
every y tick label/.append style={font=\color{black}, font=\tiny},
axis background/.style={fill=white},
axis x line*=bottom,
axis y line*=left,
xmajorgrids,
ymajorgrids,
title style={font=\bfseries, at={(0.49,0.91)}},
legend style={at={(0.97,0.03)}, anchor=south east, legend cell align=left, align=left, draw=white!15!black}
]
\addplot [color=blue, line width=2.0pt]
  table[row sep=crcr]{%
0	93\\
0.00400000000000489	93\\
0.00499999999999545	100\\
0.100999999999999	100\\
};
\addlegendentry{\footnotesize 1\%}

\addplot [color=mycolor1, line width=2.0pt]
  table[row sep=crcr]{%
0	62\\
0.00400000000000489	62\\
0.00499999999999545	77\\
0.00799999999999557	77\\
0.00900000000000034	83\\
0.0100000000000051	91\\
0.0120000000000005	91\\
0.0130000000000052	92\\
0.0139999999999958	97\\
0.0169999999999959	97\\
0.0180000000000007	98\\
0.0220000000000056	98\\
0.0229999999999961	99\\
0.0400000000000063	99\\
0.0409999999999968	100\\
0.100999999999999	100\\
};
\addlegendentry{\footnotesize 5\%}

\addplot [color=mycolor2, line width=2.0pt]
  table[row sep=crcr]{%
0	44\\
0.00400000000000489	44\\
0.00499999999999545	53\\
0.007000000000005	53\\
0.00799999999999557	54\\
0.00900000000000034	57\\
0.0100000000000051	68\\
0.0120000000000005	68\\
0.0130000000000052	69\\
0.0139999999999958	73\\
0.0150000000000006	75\\
0.0169999999999959	75\\
0.0180000000000007	78\\
0.0190000000000055	82\\
0.0220000000000056	82\\
0.0229999999999961	86\\
0.0250000000000057	88\\
0.0259999999999962	88\\
0.027000000000001	91\\
0.0280000000000058	92\\
0.0310000000000059	92\\
0.0319999999999965	94\\
0.0330000000000013	95\\
0.0360000000000014	95\\
0.0370000000000061	97\\
0.0430000000000064	97\\
0.0439999999999969	98\\
0.0460000000000065	98\\
0.046999999999997	99\\
0.0589999999999975	99\\
0.0600000000000023	100\\
0.100999999999999	100\\
};
\addlegendentry{\footnotesize 25\%}

\addplot [color=mycolor3, line width=2.0pt]
  table[row sep=crcr]{%
0	21\\
0.00400000000000489	21\\
0.00499999999999545	31\\
0.00799999999999557	31\\
0.00900000000000034	34\\
0.0100000000000051	42\\
0.0109999999999957	42\\
0.0130000000000052	44\\
0.0139999999999958	49\\
0.0150000000000006	53\\
0.0169999999999959	53\\
0.0180000000000007	58\\
0.0190000000000055	62\\
0.0210000000000008	64\\
0.0220000000000056	64\\
0.0229999999999961	71\\
0.0240000000000009	73\\
0.0259999999999962	75\\
0.027000000000001	80\\
0.0280000000000058	83\\
0.0289999999999964	84\\
0.0310000000000059	84\\
0.0319999999999965	85\\
0.0330000000000013	88\\
0.034000000000006	88\\
0.0349999999999966	89\\
0.0360000000000014	89\\
0.0379999999999967	93\\
0.0390000000000015	93\\
0.0409999999999968	95\\
0.0420000000000016	95\\
0.0430000000000064	96\\
0.0450000000000017	96\\
0.0460000000000065	97\\
0.0490000000000066	97\\
0.0510000000000019	99\\
0.0589999999999975	99\\
0.0600000000000023	100\\
0.100999999999999	100\\
};
\addlegendentry{\footnotesize 90\%}

\addplot [color=black, dotted, line width=1.5pt]
  table[row sep=crcr]{%
0	35\\
0.00400000000000489	35\\
0.00499999999999545	50\\
0.00799999999999557	50\\
0.00900000000000034	53\\
0.0100000000000051	61\\
0.0109999999999957	61\\
0.0130000000000052	63\\
0.0139999999999958	71\\
0.0150000000000006	72\\
0.0169999999999959	72\\
0.0180000000000007	76\\
0.0190000000000055	82\\
0.0210000000000008	84\\
0.0220000000000056	84\\
0.0229999999999961	88\\
0.0240000000000009	90\\
0.0259999999999962	92\\
0.027000000000001	95\\
0.0280000000000058	96\\
0.0330000000000013	96\\
0.034000000000006	97\\
0.0360000000000014	97\\
0.0370000000000061	98\\
0.0529999999999973	98\\
0.054000000000002	99\\
0.0589999999999975	99\\
0.0600000000000023	100\\
0.100999999999999	100\\
};
\addlegendentry{\footnotesize FM}

\end{axis}
\end{tikzpicture}%
\\
\end{minipage}\hspace{15mm}
\begin{minipage}{0.1\linewidth}
\centering
\setlength\figurewidth{\linewidth}
\vspace{3mm}
 \input{./1d_image_source.tikz}\\
 \small \hspace{-2mm}Source
  \vspace{15mm}
\end{minipage} \hspace{2mm}
\begin{minipage}{0.1\linewidth}
\centering
\setlength\figurewidth{\linewidth}
\input{./1d_image_fm.tikz}\vspace{5mm}\\
  \begin{overpic}
  [trim=0cm 0cm 0cm 0cm,clip,width=\linewidth]{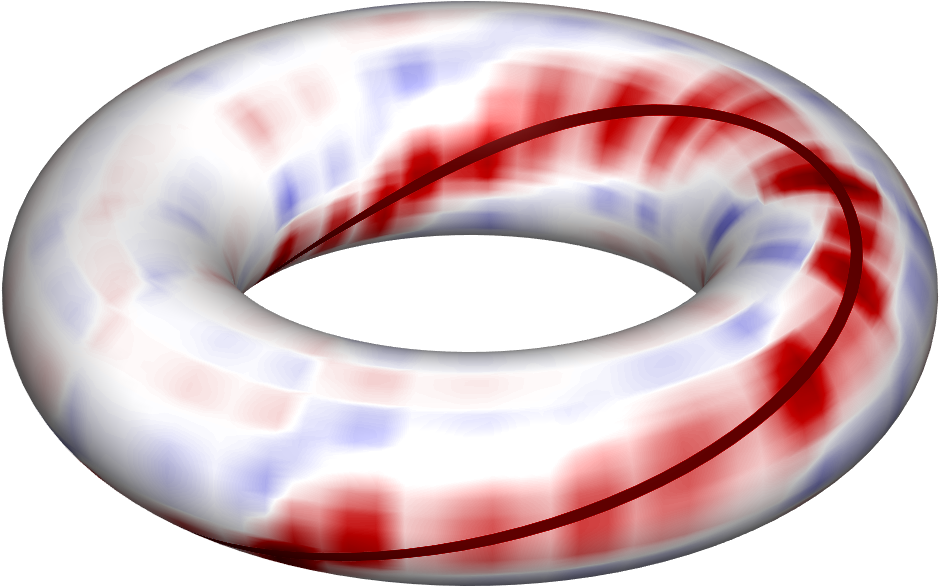}
  \end{overpic}
  \small FM
\end{minipage} \hspace{2mm}
\begin{minipage}{0.1\linewidth}
\centering
\setlength\figurewidth{\linewidth}
\input{./1d_image_90.tikz}\vspace{5mm}\\
  \begin{overpic}
  [trim=0cm 0cm 0cm 0cm,clip,width=\linewidth]{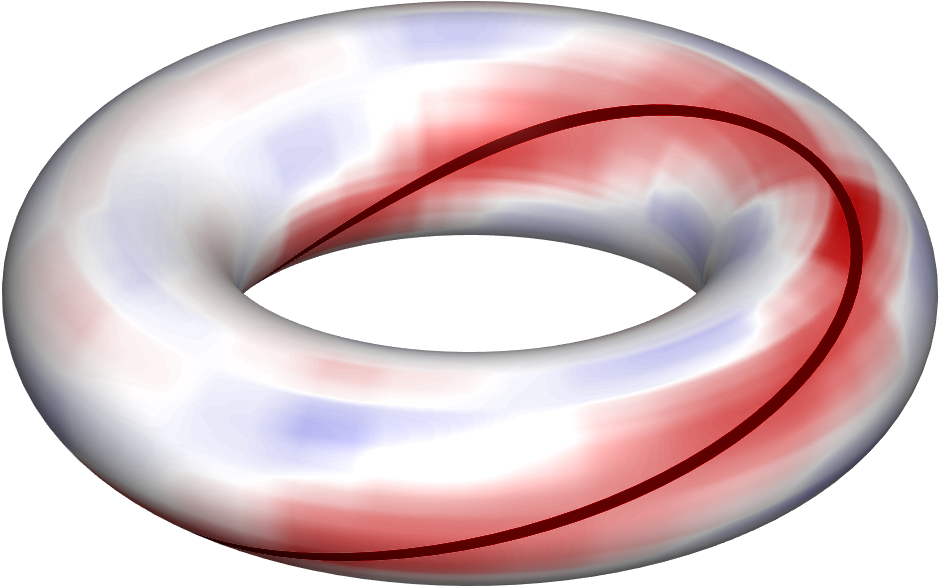}
  \end{overpic}
   \small 90\%
\end{minipage} \hspace{2mm}
\begin{minipage}{0.1\linewidth}
\centering
\setlength\figurewidth{\linewidth}
\input{./1d_image_25.tikz}\vspace{5mm}\\
  \begin{overpic}
  [trim=0cm 0cm 0cm 0cm,clip,width=\linewidth]{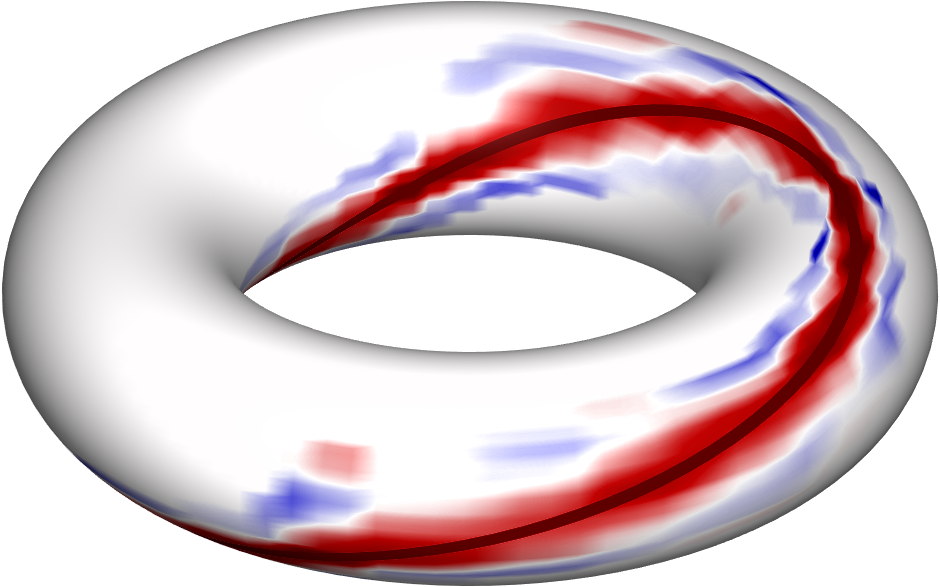}
  \end{overpic}
     \small 25\%
\end{minipage}\hspace{2mm}
\begin{minipage}{0.1\linewidth}
\centering
\setlength\figurewidth{\linewidth}
\input{./1d_image_5.tikz}\vspace{5mm}\\
  \begin{overpic}
  [trim=0cm 0cm 0cm 0cm,clip,width=\linewidth]{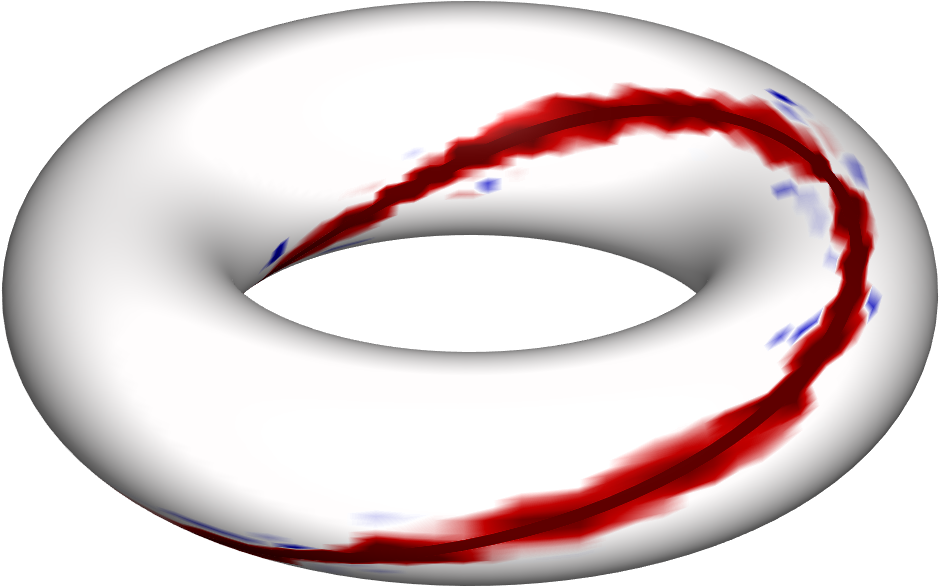}
  \end{overpic}
     \small 5\%
\end{minipage}\hspace{2mm}
\begin{minipage}{0.1\linewidth}
\centering
\setlength\figurewidth{\linewidth}
\input{./1d_image_1.tikz}\vspace{5mm}\\
  \begin{overpic}
  [trim=0cm 0cm 0cm 0cm,clip,width=\linewidth]{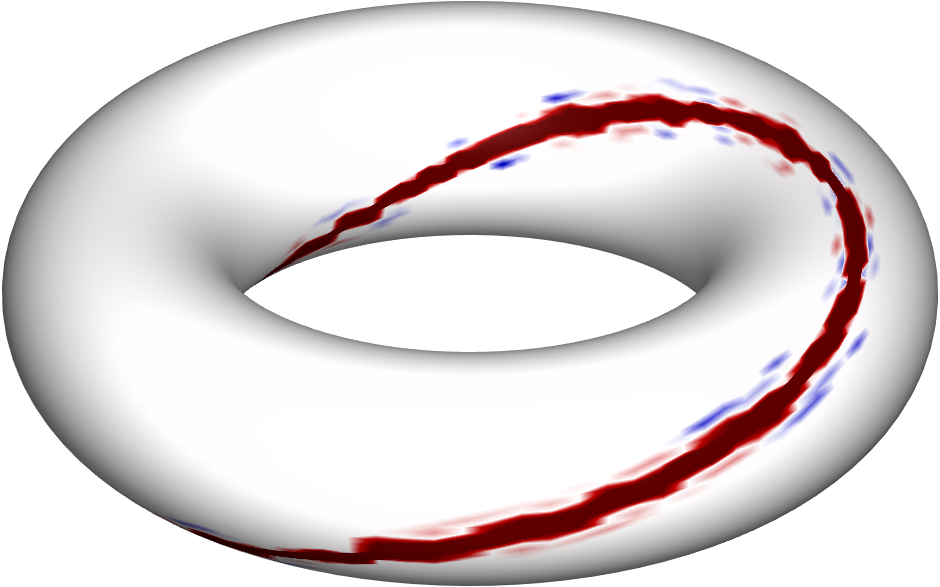}
  \end{overpic}
   \small 1\%
\end{minipage}
\vspace{-0.4cm}
  \caption{\label{fig:support} Product space approximation of the correspondence between one-dimensional shapes with $k=100$ basis functions. Bases constructed on bands of different size ($1\%$, $5\%$, $25\%$ and $90\%$ of the total product manifold area) around the true correspondence are shown. Separable basis (FM) is shown as a reference. Left: accuracy of the correspondence increases as the product space basis becomes more localized. Right (top row): image of a delta function by the functional maps. Right (bottom row): True correspondence (curve) and its approximation in inseparable product space bases with a varying degree of localization. The product manifold is depicted as a torus.}
\end{figure*}
%

%
%

%
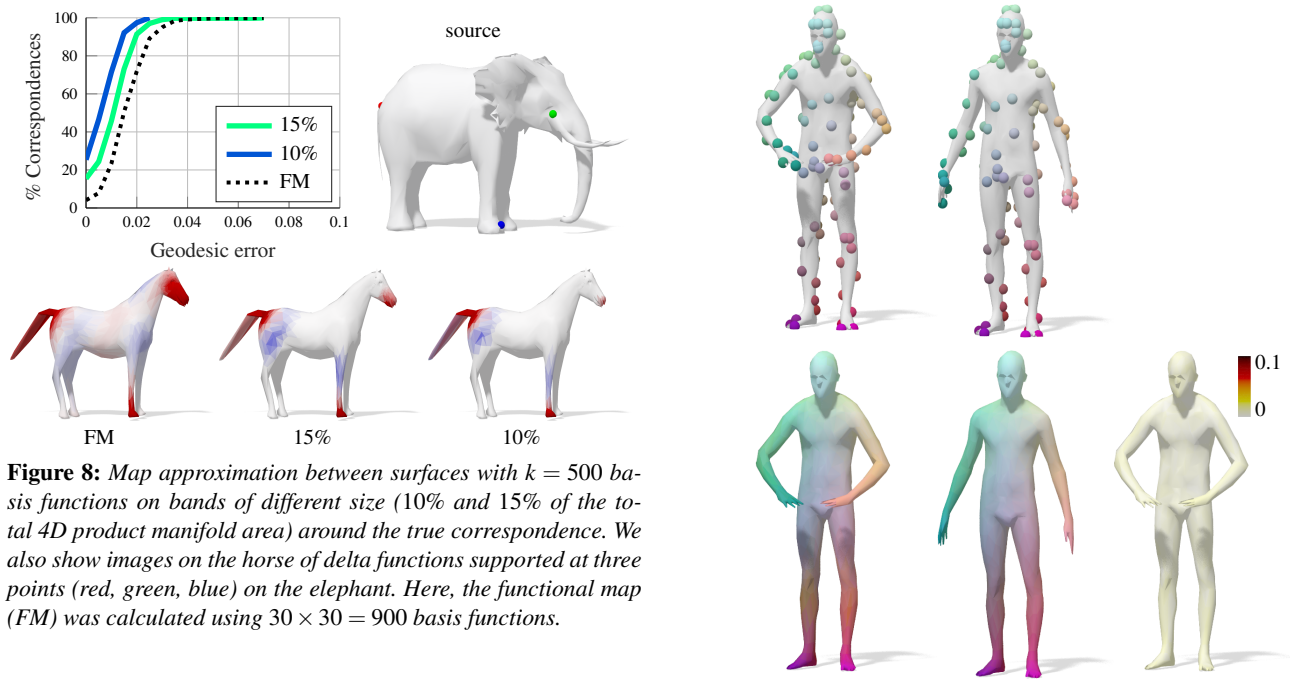
\begin{figure}[bt]
  \centering
%
%
\definecolor{mycolor1}{rgb}{0.00000,0.33333,0.83333}%
\definecolor{mycolor2}{rgb}{0.00000,0.66667,0.66667}%
\definecolor{mycolor3}{rgb}{0.00000,1.00000,0.50000}%
\pgfplotsset{scaled x ticks=false}
\begin{tikzpicture}

\begin{axis}[%
width=0.40\columnwidth,
height=0.30\columnwidth,
scale only axis,
xmin=0,
xmax=0.01,
xticklabels={0,0,0.02,0.04,0.06,0.08,0.1},
xlabel style={font=\color{white!15!black}},
xlabel style={at={(0.5,0.07)}},
xlabel={\footnotesize Geodesic error},
ymin=0,
ymax=100,
ylabel style={font=\color{white!15!black}},
ylabel style={at={(0.165,0.48)}},
ylabel={\footnotesize \% Correspondences},
every x tick label/.append style={font=\color{black}, font=\tiny},
every y tick label/.append style={font=\color{black}, font=\tiny},
axis background/.style={fill=white},
axis x line*=bottom,
axis y line*=left,
xmajorgrids,
ymajorgrids,
title style={font=\bfseries, at={(0.49,0.91)}},
legend style={at={(0.97,0.03)}, anchor=south east, legend cell align=left, align=left, draw=white!15!black}
]
\addplot [color=mycolor3,solid,line width=2.0pt]
  table[row sep=crcr]{%
0	15.4295246038365\\
0.0005	24.3536280233528\\
0.001	45.7881567973311\\
0.0015	73.0608840700584\\
0.002	91.4095079232694\\
0.0025	96.9974979149291\\
0.003	98.999165971643\\
0.0035	99.7497914929108\\
0.004	99.7497914929108\\
0.0045	99.7497914929108\\
0.005	99.7497914929108\\
0.0055	99.7497914929108\\
0.006	99.7497914929108\\
0.0065	99.7497914929108\\
0.007	99.9165971643036\\
};
\addlegendentry{\footnotesize 15\%}
\addplot [color=mycolor1,solid,line width=2.0pt]
  table[row sep=crcr]{%
0	25.3544620517098\\
0.0005	46.5387823185988\\
0.001	71.4762301918265\\
0.0015	92.1601334445371\\
0.002	97.4145120934112\\
0.0025	100\\
};
\addlegendentry{\footnotesize 10\%}
\addplot [color=black,dotted,line width=1.5pt]
  table[row sep=crcr]{%
0	4.08673894912427\\
0.0005	7.83986655546289\\
0.001	23.3527939949958\\
0.0015	50.5421184320267\\
0.002	71.8932443703086\\
0.0025	89.0742285237698\\
0.003	94.9958298582152\\
0.0035	98.4153461217681\\
0.004	99.0825688073395\\
0.0045	99.4161801501251\\
0.005	99.6663886572143\\
0.0055	99.7497914929108\\
0.006	99.7497914929108\\
0.0065	99.7497914929108\\
0.007	99.9165971643036\\
};
\addlegendentry{\footnotesize FM}
\end{axis}
\end{tikzpicture}%
  \hspace{0.1cm}
  \begin{overpic}
  [trim=0cm 0cm 0cm 0cm,clip,width=0.4\linewidth]{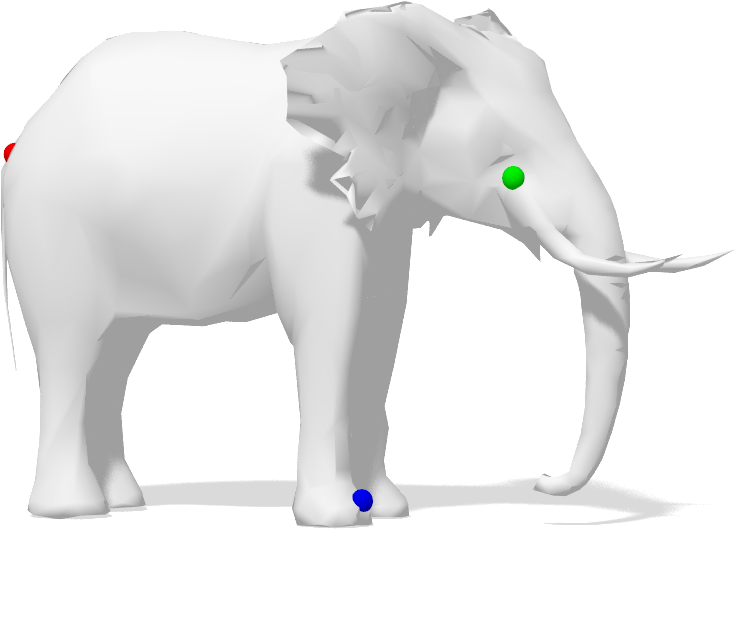}
  \put(27,90){\footnotesize source}
  \end{overpic}
  \begin{overpic}
  [trim=0cm 0cm 0cm 0cm,clip,width=0.32\linewidth]{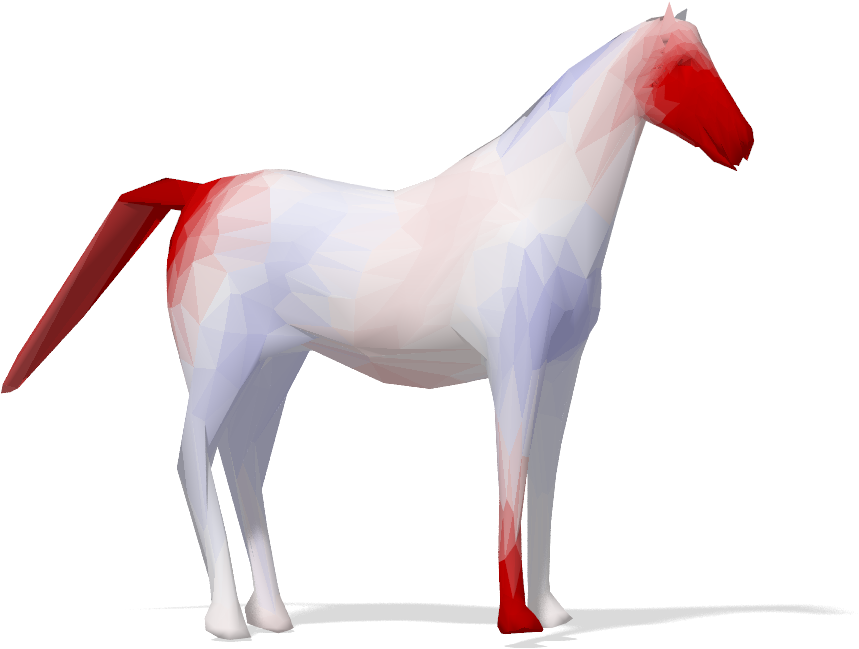}
  \put(36,-12){\footnotesize FM}
  \end{overpic}
  \begin{overpic}
  [trim=0cm 0cm 0cm 0cm,clip,width=0.32\linewidth]{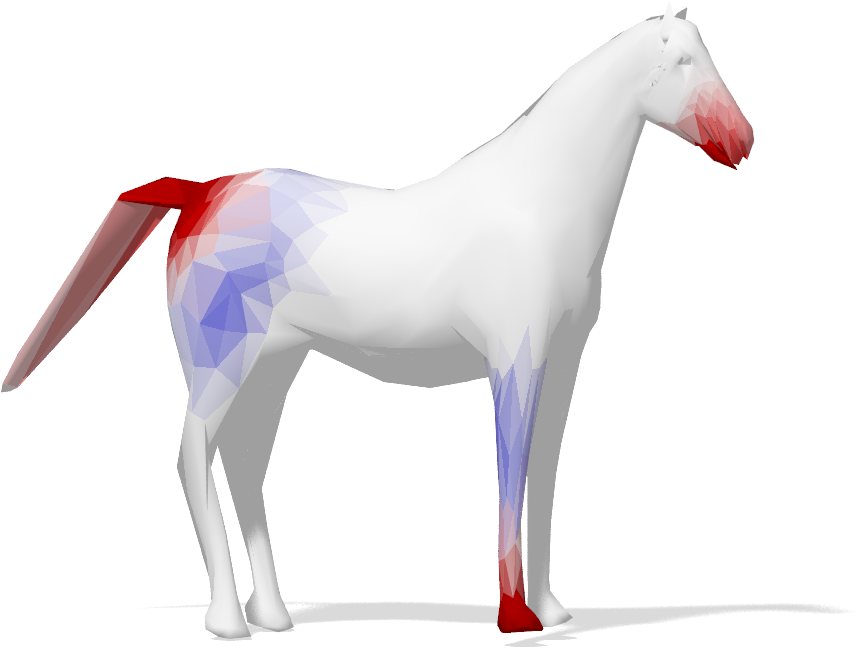}
  \put(36,-12){\footnotesize $15\%$}
  \end{overpic}
  \begin{overpic}
  [trim=0cm 0cm 0cm 0cm,clip,width=0.32\linewidth]{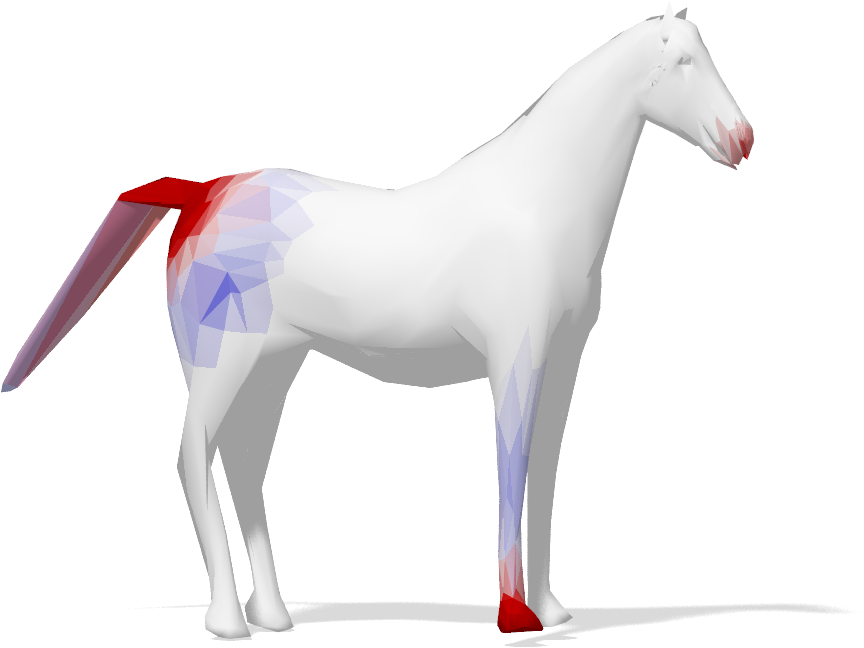}
  \put(36,-12){\footnotesize $10\%$}
  \end{overpic}
  \vspace{0.2cm}
  \caption{\label{fig:delta3d}Map approximation between surfaces with $k=500$ basis functions on bands of different size ($10\%$ and $15\%$ of the total 4D product manifold area) around the true correspondence. We also show images on the horse of delta functions supported at three points (red, green, blue) on the elephant. Here, the functional map (FM) was calculated using $30\times30=900$ basis functions.}
\end{figure}

\noindent\textbf{Example: Map refinement.}
As an illustrative application of our framework, we propose a simple procedure for \rev{map refinement: Given some initial, possibly sparse and noisy correspondence, the task is to produce a dense, denoised map}.



\rev{We follow an iterative approach. In each iteration $k$, the map is represented as a density $\mu^{(k)} : \M\times\N \to [0,1]$. This density is interpreted as a heat distribution throughout the iterations.

At the $k$-th iteration, a diffusion process is initialized with $u^{(k)}_{t=0}:=\mu^{(k)}$ and solved for a given diffusion time $T^{(k)}$ \final{on a patch $\mathcal{P}^{(k)}$. The initial patch can be given or be the entire product manifold}. The diffusion process has the effect of spreading correct correspondence information and therefore suppress mismatches, resulting in an effective map denoising approach akin to diffusion-based smoothing from image processing~\cite{scalespace83,peronamalik90}. The final heat distribution $u^{(k)}_{T}$ is thresholded to define a patch $\mathcal{P}^{(k)}\subset\M\times\N$ where the correct correspondence is likely to be contained, with likelihood expressed in terms of the diffused density.
We then recover a bijective (non-soft) density $\mu^{(k+1)}$ from $u^{(k)}_{T}$ by solving a linear assignment problem~\cite{bertsekas98} restricted to region $\mathcal{P}^{(k)}$, and use it to initialize the next iteration. 

These blur-and-sharpen steps are iterated until convergence while decreasing $T^{(k)}$, resulting in a sequence $\mathcal{P}^{(0)} \supseteq \cdots \supseteq \mathcal{P}^{(k)} \supseteq \mathcal{P}^{(k+1)}$. In practice, we decrease $T^{(k)}$  logarithmically across iterations. At $k=0$, the density $u^{(0)}$ is the given input, e.g., a mixture of Dirac deltas or a soft map.
%

The diffusion step in each iteration is realized via the spectral decomposition of the product patch Laplacian $\Delta_\mathcal{P}$ with Dirichlet boundary conditions \final{on $\mathcal{P} = \mathcal{P}^{(k)}$} \eqref{eq:patchbc}; for $p,q \in\M \times \N$:}
\begin{align}
u_T(p) &= \int_{\M\times\N} h_T (p,q) u_0(q) \mathrm{d}q\\
h_T(p,q) &= \sum_{\ell\ge 0} e^{-T \gamma_\ell} \bar{\xi}_\ell(p) \bar{\xi}_\ell(q)\label{eq:htpq}\,,
\end{align}
\rev{where $h_T$ is the heat kernel at time $T$ on the product manifold ${\M\times\N}$. Throughout the iterations we keep the number of eigenfunctions for the approximation \eqref{eq:htpq} constant.

The refinement process described above simultaneously improves the correspondence and reduces the support of the density around the most likely bijective map. 
%
%
This is similar in spirit to the kernel matching approaches of \cite{vestner2017product,vestnerefficient}, however, with the additional step of `carving out' the relevant portion $\mathcal{P}\subset\M\times\N$ throughout the iterations.

Illustrative results are reported qualitatively in Figure~\ref{fig:denoising} and quantitatively in Figure~\ref{fig:q1q2}.
}

\begin{figure}[bt]
  \centering
  \begin{overpic}
  [trim=0cm 0cm 0cm 0cm,clip,width=0.85\linewidth]{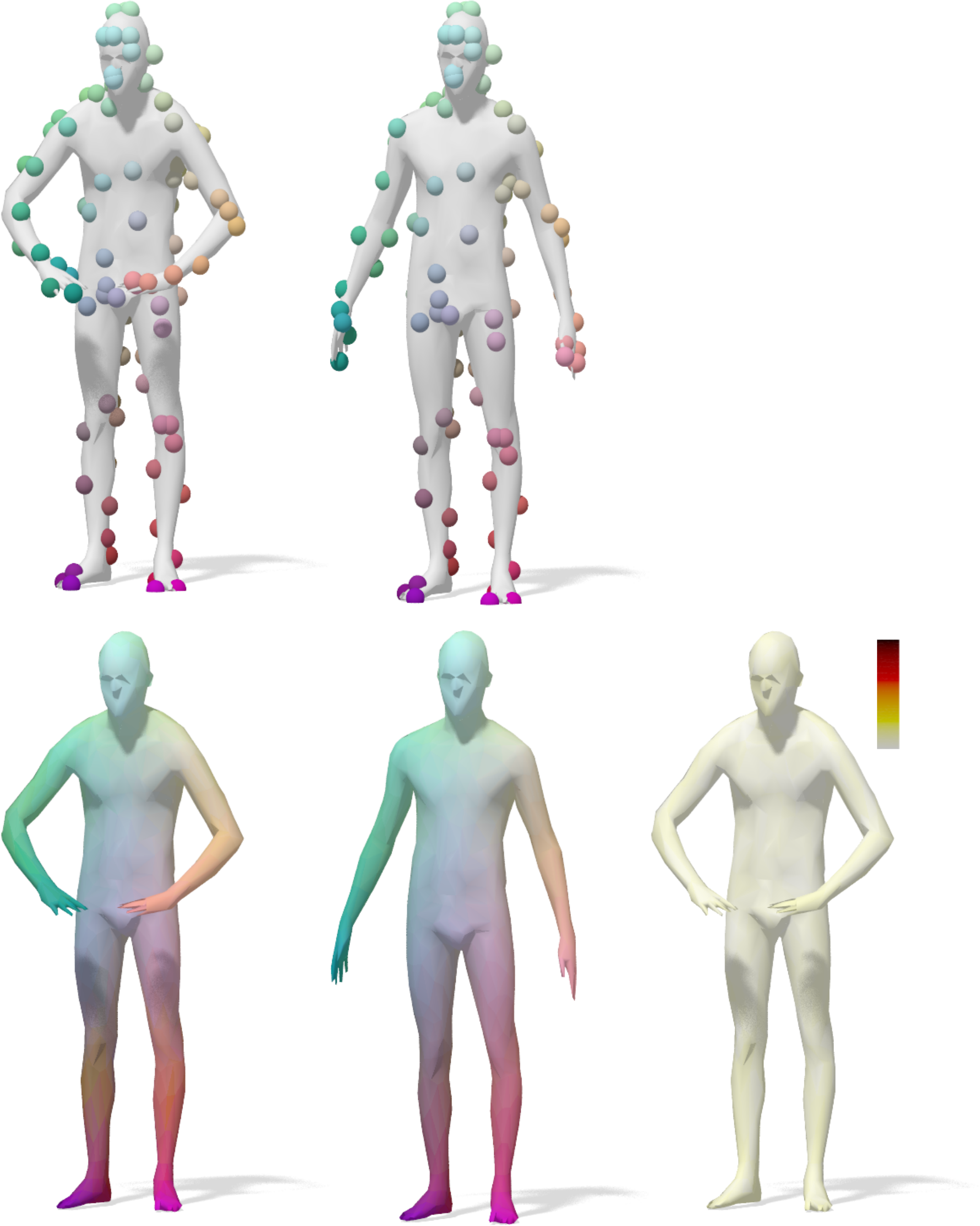}
  \put(74,46){\footnotesize 0.1}
  \put(74,39){\footnotesize 0}
  \end{overpic}
  \caption{\label{fig:denoising}Example of map refinement. We show the input sparse correspondence above and the recovered dense map below. The heatmap on the bottom right encodes geodesic error of the recovered correspondence.}
\end{figure}

\begin{figure}
  \centering
%
%
\definecolor{mycolor1}{rgb}{0.00000,0.44700,0.74100}%
\definecolor{mycolor2}{rgb}{0.85000,0.32500,0.09800}%
\definecolor{mycolor3}{rgb}{0.92900,0.69400,0.12500}%
\definecolor{mycolor4}{rgb}{0.49400,0.18400,0.55600}%
\definecolor{mycolor5}{rgb}{0.46600,0.67400,0.18800}%
\definecolor{mycolor6}{rgb}{0.30100,0.74500,0.93300}%
\definecolor{mycolor7}{rgb}{0.63500,0.07800,0.18400}%
\begin{tikzpicture}

\begin{axis}[%
width=.38\linewidth,
height=.5\linewidth,
scale only axis,
xmin=0,
xmax=0.1,
ymin=0,
ymax=100,
axis background/.style={fill=white},
xticklabels={0,0,0.02,0.04,0.06,0.08,0.1},
xlabel style={font=\color{white!15!black}},
xlabel style={at={(0.5,0.05)}},
xlabel={\footnotesize Geodesic error},
ylabel style={font=\color{white!15!black}},
ylabel style={at={(0.19,0.48)}},
ylabel={\footnotesize \% Correspondences},
every x tick label/.append style={font=\color{black}, font=\tiny},
every y tick label/.append style={font=\color{black}, font=\tiny},
axis background/.style={fill=white},
axis x line*=bottom,
axis y line*=left,
xmajorgrids,
ymajorgrids,
title style={font=\bfseries, at={(0.49,0.91)}},
legend style={at={(0.97,0.03)}, anchor=south east, legend cell align=left, align=left, draw=white!15!black, font=\tiny}
]


\addplot [color=mycolor2,solid,line width=1.5pt]
  table[row sep=crcr]{%
0	29.3\\
0.005	29.8\\
0.01	35.9\\
0.015	49.6\\
0.02	56.9\\
0.025	62.9\\
0.03	68\\
0.035	72.7\\
0.04	76.8\\
0.045	79.9\\
0.05	83.5\\
0.055	86.5\\
0.06	89\\
0.065	91.7\\
0.07	93.4\\
0.075	95.1\\
0.08	96.3\\
0.085	96.6\\
0.09	97.1\\
0.095	97.8\\
0.1	97.9\\
0.105	98.1\\
0.11	98.2\\
0.115	98.3\\
0.12	98.5\\
0.125	99.2\\
0.13	99.5\\
0.135	99.8\\
0.14	99.9\\
0.145	99.9\\
0.15	100\\
0.155	100\\
0.16	100\\
0.165	100\\
0.17	100\\
0.175	100\\
0.18	100\\
0.185	100\\
0.19	100\\
0.195	100\\
0.2	100\\
0.205	100\\
0.21	100\\
0.215	100\\
0.22	100\\
0.225	100\\
0.23	100\\
0.235	100\\
0.24	100\\
0.245	100\\
0.25	100\\
};
\addlegendentry{$(0.1,0.01)$};

\addplot [color=mycolor3,solid,line width=1.5pt]
  table[row sep=crcr]{%
0	31\\
0.005	31.4\\
0.01	40\\
0.015	53.8\\
0.02	61.8\\
0.025	68.2\\
0.03	71.9\\
0.035	76.2\\
0.04	80\\
0.045	83.6\\
0.05	86.2\\
0.055	89.5\\
0.06	91.8\\
0.065	93.5\\
0.07	94.7\\
0.075	95.9\\
0.08	97.6\\
0.085	98.2\\
0.09	98.5\\
0.095	98.6\\
0.1	98.8\\
0.105	99\\
0.11	99.4\\
0.115	99.5\\
0.12	99.6\\
0.125	99.9\\
0.13	99.9\\
0.135	100\\
0.14	100\\
0.145	100\\
0.15	100\\
0.155	100\\
0.16	100\\
0.165	100\\
0.17	100\\
0.175	100\\
0.18	100\\
0.185	100\\
0.19	100\\
0.195	100\\
0.2	100\\
0.205	100\\
0.21	100\\
0.215	100\\
0.22	100\\
0.225	100\\
0.23	100\\
0.235	100\\
0.24	100\\
0.245	100\\
0.25	100\\
};
\addlegendentry{$(1,0.01)$};

\addplot [color=mycolor4,solid,line width=1.5pt]
  table[row sep=crcr]{%
0	33.8\\
0.005	34.5\\
0.01	43.3\\
0.015	57.8\\
0.02	65.7\\
0.025	70.8\\
0.03	75.8\\
0.035	79.4\\
0.04	82.7\\
0.045	85.2\\
0.05	87.8\\
0.055	90.5\\
0.06	93.1\\
0.065	94.4\\
0.07	96\\
0.075	97\\
0.08	98.3\\
0.085	98.8\\
0.09	99\\
0.095	99.4\\
0.1	99.4\\
0.105	99.8\\
0.11	99.8\\
0.115	99.8\\
0.12	99.9\\
0.125	99.9\\
0.13	99.9\\
0.135	100\\
0.14	100\\
0.145	100\\
0.15	100\\
0.155	100\\
0.16	100\\
0.165	100\\
0.17	100\\
0.175	100\\
0.18	100\\
0.185	100\\
0.19	100\\
0.195	100\\
0.2	100\\
0.205	100\\
0.21	100\\
0.215	100\\
0.22	100\\
0.225	100\\
0.23	100\\
0.235	100\\
0.24	100\\
0.245	100\\
0.25	100\\
};
\addlegendentry{$(10,0.1)$};

\addplot [color=mycolor5,solid,line width=1.5pt]
  table[row sep=crcr]{%
0	35.6\\
0.005	36\\
0.01	44\\
0.015	59.1\\
0.02	69.1\\
0.025	74.3\\
0.03	78.9\\
0.035	82.6\\
0.04	85.3\\
0.045	88.3\\
0.05	90.3\\
0.055	92\\
0.06	93.8\\
0.065	95.4\\
0.07	96.7\\
0.075	97.5\\
0.08	98.6\\
0.085	99.1\\
0.09	99.6\\
0.095	99.7\\
0.1	99.7\\
0.105	99.8\\
0.11	99.8\\
0.115	99.8\\
0.12	99.8\\
0.125	99.9\\
0.13	99.9\\
0.135	100\\
0.14	100\\
0.145	100\\
0.15	100\\
0.155	100\\
0.16	100\\
0.165	100\\
0.17	100\\
0.175	100\\
0.18	100\\
0.185	100\\
0.19	100\\
0.195	100\\
0.2	100\\
0.205	100\\
0.21	100\\
0.215	100\\
0.22	100\\
0.225	100\\
0.23	100\\
0.235	100\\
0.24	100\\
0.245	100\\
0.25	100\\
};
\addlegendentry{$(100,1)$};

\addplot [color=mycolor6,solid,line width=1.5pt]
  table[row sep=crcr]{%
0	10.6\\
0.005	10.7\\
0.01	14.8\\
0.015	21.7\\
0.02	29.4\\
0.025	37.4\\
0.03	43.5\\
0.035	48.7\\
0.04	53.7\\
0.045	58.5\\
0.05	62.7\\
0.055	65.8\\
0.06	68.8\\
0.065	72.8\\
0.07	75.3\\
0.075	77.4\\
0.08	79.7\\
0.085	82.2\\
0.09	83.6\\
0.095	84.8\\
0.1	85.5\\
0.105	86.8\\
0.11	88.2\\
0.115	89.5\\
0.12	90.5\\
0.125	90.9\\
0.13	91.7\\
0.135	92\\
0.14	92.5\\
0.145	93\\
0.15	93.6\\
0.155	94.1\\
0.16	94.3\\
0.165	94.9\\
0.17	95.2\\
0.175	95.6\\
0.18	95.8\\
0.185	96.1\\
0.19	96.3\\
0.195	96.3\\
0.2	96.3\\
0.205	96.6\\
0.21	96.7\\
0.215	96.8\\
0.22	97.2\\
0.225	97.4\\
0.23	97.5\\
0.235	97.5\\
0.24	97.5\\
0.245	97.5\\
0.25	97.5\\
};
\addlegendentry{$(1000,0.01)$};

\addplot [color=mycolor7,solid,line width=1.5pt]
  table[row sep=crcr]{%
0	10.4\\
0.005	10.6\\
0.01	15.9\\
0.015	21\\
0.02	26.6\\
0.025	33.2\\
0.03	38.7\\
0.035	44.3\\
0.04	47.6\\
0.045	51.7\\
0.05	55.3\\
0.055	59.1\\
0.06	61.6\\
0.065	64.5\\
0.07	68.1\\
0.075	69.9\\
0.08	71.8\\
0.085	73.7\\
0.09	75\\
0.095	76.7\\
0.1	78.6\\
0.105	80.3\\
0.11	81.7\\
0.115	82.9\\
0.12	83.5\\
0.125	84.6\\
0.13	85.7\\
0.135	86.4\\
0.14	86.9\\
0.145	88\\
0.15	88.7\\
0.155	90.3\\
0.16	91.3\\
0.165	92.1\\
0.17	93.1\\
0.175	94.3\\
0.18	95.9\\
0.185	97.2\\
0.19	97.8\\
0.195	98.7\\
0.2	99.3\\
0.205	99.8\\
0.21	99.9\\
0.215	100\\
0.22	100\\
0.225	100\\
0.23	100\\
0.235	100\\
0.24	100\\
0.245	100\\
0.25	100\\
};
\addlegendentry{$(1000,10)$};

\end{axis}
\end{tikzpicture}%
%
%
\definecolor{mycolor1}{rgb}{0.00000,0.44700,0.74100}%
\definecolor{mycolor2}{rgb}{0.85000,0.32500,0.09800}%
\definecolor{mycolor3}{rgb}{0.92900,0.69400,0.12500}%
\definecolor{mycolor4}{rgb}{0.49400,0.18400,0.55600}%
\definecolor{mycolor5}{rgb}{0.46600,0.67400,0.18800}%
\definecolor{mycolor6}{rgb}{0.30100,0.74500,0.93300}%
\definecolor{mycolor7}{rgb}{0.63500,0.07800,0.18400}%
\begin{tikzpicture}

\begin{axis}[%
width=.38\linewidth,
height=.5\linewidth,
scale only axis,
xmin=0,
xmax=0.1,
ymin=0,
ymax=100,
axis background/.style={fill=white},
xticklabels={0,0,0.02,0.04,0.06,0.08,0.1},
xlabel style={font=\color{white!15!black}},
xlabel style={at={(0.5,0.05)}},
xlabel={\footnotesize Geodesic error},
ylabel style={font=\color{white!15!black}},
ylabel style={at={(0.09,0.48)}},
every x tick label/.append style={font=\color{black}, font=\tiny},
every y tick label/.append style={font=\color{black}, font=\tiny},
axis background/.style={fill=white},
axis x line*=bottom,
axis y line*=left,
xmajorgrids,
ymajorgrids,
title style={font=\bfseries, at={(0.49,0.91)}},
legend style={at={(0.97,0.03)}, anchor=south east, legend cell align=left, align=left, draw=white!15!black, font=\tiny}
]

\addplot [color=mycolor2,solid,line width=1.5pt]
  table[row sep=crcr]{%
0	19.828\\
0.005	20.104\\
0.01	25.335\\
0.015	34.775\\
0.02	42.271\\
0.025	48.418\\
0.03	53.865\\
0.035	58.383\\
0.04	62.357\\
0.045	66.236\\
0.05	69.61\\
0.055	72.693\\
0.06	75.204\\
0.065	77.573\\
0.07	79.655\\
0.075	81.318\\
0.08	82.899\\
0.085	83.938\\
0.09	84.87\\
0.095	85.622\\
0.1	86.26\\
0.105	86.982\\
0.11	87.613\\
0.115	88.195\\
0.12	88.774\\
0.125	89.3159999999999\\
0.13	89.872\\
0.135	90.288\\
0.14	90.749\\
0.145	91.09\\
0.15	91.37\\
0.155	91.716\\
0.16	92.017\\
0.165	92.304\\
0.17	92.596\\
0.175	92.896\\
0.18	93.111\\
0.185	93.329\\
0.19	93.508\\
0.195	93.681\\
0.2	93.826\\
0.205	93.949\\
0.21	94.055\\
0.215	94.17\\
0.22	94.25\\
0.225	94.322\\
0.23	94.397\\
0.235	94.449\\
0.24	94.5\\
0.245	94.56\\
0.25	94.616\\
};
\addlegendentry{0.1-0.01};

\addplot [color=mycolor3,solid,line width=1.5pt]
  table[row sep=crcr]{%
0	22.099\\
0.005	22.376\\
0.01	27.786\\
0.015	37.924\\
0.02	45.627\\
0.025	51.642\\
0.03	56.852\\
0.035	61.215\\
0.04	65.111\\
0.045	68.858\\
0.05	72.052\\
0.055	75.013\\
0.06	77.51\\
0.065	79.715\\
0.07	81.598\\
0.075	83.088\\
0.08	84.66\\
0.085	85.5820000000001\\
0.09	86.43\\
0.095	87.035\\
0.1	87.595\\
0.105	88.128\\
0.11	88.634\\
0.115	89.137\\
0.12	89.675\\
0.125	90.152\\
0.13	90.6\\
0.135	90.984\\
0.14	91.284\\
0.145	91.562\\
0.15	91.836\\
0.155	92.095\\
0.16	92.344\\
0.165	92.598\\
0.17	92.843\\
0.175	93.092\\
0.18	93.267\\
0.185	93.474\\
0.19	93.619\\
0.195	93.762\\
0.2	93.885\\
0.205	93.979\\
0.21	94.058\\
0.215	94.137\\
0.22	94.191\\
0.225	94.242\\
0.23	94.294\\
0.235	94.337\\
0.24	94.385\\
0.245	94.414\\
0.25	94.481\\
};
\addlegendentry{1-0.01};

\addplot [color=mycolor4,solid,line width=1.5pt]
  table[row sep=crcr]{%
0	24.536\\
0.005	24.904\\
0.01	30.859\\
0.015	41.611\\
0.02	49.438\\
0.025	55.515\\
0.03	60.8\\
0.035	65.2019999999999\\
0.04	69.031\\
0.045	72.964\\
0.05	76.061\\
0.055	78.936\\
0.06	81.385\\
0.065	83.665\\
0.07	85.583\\
0.075	87.046\\
0.08	88.444\\
0.085	89.368\\
0.09	90.178\\
0.095	90.768\\
0.1	91.208\\
0.105	91.656\\
0.11	92.04\\
0.115	92.459\\
0.12	92.935\\
0.125	93.3\\
0.13	93.653\\
0.135	93.927\\
0.14	94.146\\
0.145	94.334\\
0.15	94.531\\
0.155	94.752\\
0.16	94.93\\
0.165	95.137\\
0.17	95.313\\
0.175	95.474\\
0.18	95.579\\
0.185	95.727\\
0.19	95.8699999999999\\
0.195	95.977\\
0.2	96.093\\
0.205	96.182\\
0.21	96.262\\
0.215	96.327\\
0.22	96.379\\
0.225	96.429\\
0.23	96.48\\
0.235	96.508\\
0.24	96.551\\
0.245	96.599\\
0.25	96.644\\
};
\addlegendentry{10-0.1};

\addplot [color=mycolor5,solid,line width=1.5pt]
  table[row sep=crcr]{%
0	35.117\\
0.005	35.684\\
0.01	43.615\\
0.015	57.399\\
0.02	65.99\\
0.025	71.751\\
0.03	76.192\\
0.035	80.013\\
0.04	83.2649999999999\\
0.045	86.498\\
0.05	88.91\\
0.055	91.335\\
0.06	93.336\\
0.065	95.055\\
0.07	96.539\\
0.075	97.558\\
0.08	98.519\\
0.085	99.091\\
0.09	99.4910000000001\\
0.095	99.582\\
0.1	99.642\\
0.105	99.7589999999999\\
0.11	99.7839999999999\\
0.115	99.8299999999999\\
0.12	99.8929999999999\\
0.125	99.9329999999999\\
0.13	99.9549999999999\\
0.135	99.998\\
0.14	99.998\\
0.145	99.998\\
0.15	99.998\\
0.155	99.998\\
0.16	99.998\\
0.165	99.998\\
0.17	99.998\\
0.175	99.998\\
0.18	99.998\\
0.185	99.998\\
0.19	99.998\\
0.195	99.998\\
0.2	99.998\\
0.205	99.998\\
0.21	99.998\\
0.215	99.998\\
0.22	99.998\\
0.225	99.998\\
0.23	99.998\\
0.235	99.998\\
0.24	99.998\\
0.245	99.998\\
0.25	99.998\\
};
\addlegendentry{100-1};

\addplot [color=mycolor6,solid,line width=1.5pt]
  table[row sep=crcr]{%
0	9.283\\
0.005	9.585\\
0.01	13.703\\
0.015	19.156\\
0.02	24.658\\
0.025	29.577\\
0.03	34.259\\
0.035	38.67\\
0.04	42.516\\
0.045	46.453\\
0.05	49.704\\
0.055	53.034\\
0.06	55.902\\
0.065	58.52\\
0.07	61.125\\
0.075	63.042\\
0.08	65.024\\
0.085	66.631\\
0.09	68.14\\
0.095	69.384\\
0.1	70.602\\
0.105	71.857\\
0.11	73.096\\
0.115	74.304\\
0.12	75.526\\
0.125	76.693\\
0.13	77.871\\
0.135	78.868\\
0.14	79.865\\
0.145	80.8\\
0.15	81.738\\
0.155	82.612\\
0.16	83.421\\
0.165	84.199\\
0.17	84.931\\
0.175	85.729\\
0.18	86.511\\
0.185	87.237\\
0.19	87.897\\
0.195	88.545\\
0.2	89.203\\
0.205	89.782\\
0.21	90.301\\
0.215	90.819\\
0.22	91.245\\
0.225	91.669\\
0.23	92.079\\
0.235	92.456\\
0.24	92.789\\
0.245	93.141\\
0.25	93.428\\
};
\addlegendentry{1000-0.01};

\addplot [color=mycolor7,solid,line width=1.5pt]
  table[row sep=crcr]{%
0	12.55\\
0.005	12.859\\
0.01	17.848\\
0.015	24.936\\
0.02	31.634\\
0.025	37.3\\
0.03	42.382\\
0.035	47.129\\
0.04	51.29\\
0.045	55.649\\
0.05	58.985\\
0.055	62.335\\
0.06	65.246\\
0.065	68.062\\
0.07	70.782\\
0.075	72.784\\
0.08	74.74\\
0.085	76.441\\
0.09	77.896\\
0.095	79.179\\
0.1	80.411\\
0.105	81.538\\
0.11	82.663\\
0.115	83.801\\
0.12	84.903\\
0.125	85.939\\
0.13	87.045\\
0.135	87.947\\
0.14	89.03\\
0.145	89.884\\
0.15	90.594\\
0.155	91.407\\
0.16	92.006\\
0.165	92.688\\
0.17	93.354\\
0.175	94.026\\
0.18	94.7959999999999\\
0.185	95.401\\
0.19	95.962\\
0.195	96.439\\
0.2	96.862\\
0.205	97.221\\
0.21	97.4989999999999\\
0.215	97.742\\
0.22	97.914\\
0.225	98.0739999999999\\
0.23	98.185\\
0.235	98.271\\
0.24	98.34\\
0.245	98.378\\
0.25	98.43\\
};
\addlegendentry{\tiny $t\in[1000,10]$};
\legend{}
\end{axis}
\end{tikzpicture}%
  \caption{\label{fig:q1q2}Sensitivity of map refinement to heat diffusion times and noisy input. The legend reports diffusion time ranges $(t_\mathrm{max},t_\mathrm{min})$; within each range, time is decreased logarithmically over iterations. {\em Left:} The input is a sparse correspondence of $10\%$ of correct matches. We see that high diffusion times are detrimental due to the excessive spread of correspondence information. {\em Right:} The input sparse correspondence is further corrupted with $30\%$ random mismatches.}
\end{figure}
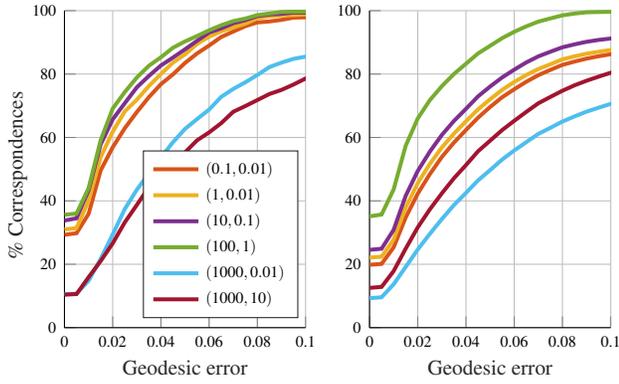

\section{Discussion and conclusions}
We introduced a novel perspective on map representation and processing, where pointwise, functional, and soft maps can be understood as densities on the product of the input shapes. We showed how to discretize the Laplace-Beltrami operator on the product manifold and proposed the adoption of (inseparable) localized harmonics for compactly encoding correspondences while ensuring minimal energy dispersion\rev{, i.e., the resulting harmonics are not `wasted' on portions of the product manifold that carry no information on the map to be encoded}. Our theoretical and applied contributions suggest a new perspective on properties of the correspondence manifold as well as new ways to pose algorithmic design for map inference and processing. 

\noindent\textbf{Limitations.}
Perhaps the main limitation of our framework lies in \rev{the scalability of our current numerical scheme}. While we showed that one can reduce the computational complexity to $\mathcal{O}(n)$ by appropriately selecting a localization region, \rev{in practical applications involving very noisy maps where the localization region tends to be spread out across the entire product manifold, the advantage might be less evident. For this reason,} considering as a possible extension higher-dimensional products to encode cycle-consistent maps in shape collections may soon become prohibitive. With the current approach we trade off scalability for accuracy: Maps are encoded much more precisely in the localized basis, but this requires the explicit computation of inseparable basis functions that do not admit an efficient representation in terms of outer products. \rev{As a possible remedy, an efficient solution to the eigenproblem might be sought via approximation methods similar to~\cite{klaus18}.}
A second limitation is in our map refinement scheme, which has limited resilience to particularly noisy input. We presented our algorithm as an illustrative tool for map denoising, but more effective schemes operating on the product manifold are likely possible.

\noindent\textbf{Future work.}
%
%
From an investigative standpoint, it might be worth considering a notion of optimal transport {\em between maps} as a means of exploring the space of maps between given shapes, a natural choice given our modeling of maps as measures on a manifold. Related constructions could extend distortion measures like the Dirichlet energy~\cite{brenier2003extended,solomon2013dirichlet,lavenant2017harmonic} to the functional regime.

Another particularly interesting direction will be to consider general {\em graphs} (as opposed to manifolds) and their products in the context of network analysis, machine learning, and applications. While many of our results may be directly translated to graphs, the lack of differentiable structure poses new theoretical challenges and at the same time provides a richer spectrum of possibilities; for example, several different notions of product exist between graphs~\cite{hammack11}.

Finally, a promising direction is the introduction of product spaces within geometric deep learning~\cite{gdl} pipelines, where the data is in the form of signals defined on top of a manifold. Our proposed discretization of the (product) Laplace-Beltrami operator, as well as its spectral decomposition, can be directly employed in such pipelines, enabling new forms of structured prediction in a range of challenging problems in vision and graphics.

{\small \section*{Acknowledgments}
\vspace{-0.1cm}
We gratefully acknowledge Mathieu Andreux, Matthias Vestner, Michael Moeller, Maks Ovsjanikov, and Paolo Rodol\`a for fruitful discussions. ER is supported by the ERC grant no.\ 802554 (SPECGEO). AB is supported by the ERC grant no.\ 335491 (RAPID). MB is supported in part by the ERC grant no.\ 724228 (LEMAN), Royal Society Wolfson Research Merit Award, Google Research Faculty Awards, Amazon AWS Machine Learning Research Award, and the Rudolf Diesel fellowship at the Institute for Advanced Studies, TU Munich. 
JS acknowledges the generous support of Army Research Office grant W911NF-12-R-0011 (``Smooth Modeling of Flows on Graphs''), of National Science Foundation grant IIS-1838071 (``BIGDATA:F:Statistical and Computational Optimal Transport for Geometric Data Analysis''), from the MIT Research Support Committee, from an Amazon Research Award, from the MIT--IBM Watson AI Laboratory, and from the Skoltech--MIT Next Generation Program.}

\appendix
\section{Proofs}\label{sec:proofs}
We provide proofs for the main propositions of the paper.

\noindent\textbf{Proof of Theorem~\ref{thm:LB1D}.}
Following standard FEM, we discretize 
the Poisson equation $\Delta_{\M\times\N}f=g$ via 
the weak formulation
\begin{equation}\label{eq:fem}
\langle \Delta_{\M\times\N} f, H_j \rangle = \langle g, H_j\rangle\,,
\end{equation}
where functions are expressed in the hat basis $\{H_j:\M\times\N\to\mathbb{R}\}$, and are thus approximated piecewise-linearly via the expansion
$
f(x) \approx \sum_{i=1}^n f(v_i) h_i(x)
$.
The left-hand side of \eqref{eq:fem} can be written as
\begin{align}
\langle \Delta f, H_j\rangle 
= -\langle \nabla f, \nabla H_j \rangle 
=- \sum_i f(v_i) \underbrace{\langle \nabla H_i, \nabla H_j \rangle}_{w_{ij}}\,,
\end{align}
where $w_{ij}$ are elements of the {\em stiffness matrix} $\mathbf{W}$. The right-hand side of \eqref{eq:fem} can be written as
\begin{align}
\langle g, H_j \rangle = \langle \sum_i g(v_i) H_i(x) , H_j \rangle = \sum_i g(v_i) \underbrace{\langle H_i, H_j \rangle}_{s_{ij}}\,,
\end{align}
where $s_{ij}$ are elements of the {\em mass matrix} $\mathbf{S}$.

The Cartesian product of the two graphs discretizing $\M$ and $\N$ has grid topology, as illustrated in Figure~\ref{fig:quad}, and the resulting bilinear hat basis functions are expressed via the outer product $H_e=h_j \wedge h_q$. We can then compute the mass values (refer to the Figure for the color notation):
\begin{fleqn}
\begin{align}
{s_{ee}} &= {\color{gray}\langle {H_e}, {H_e} \rangle} = {\color{gray}\langle} {\color{red}h_j}\wedge {\color{blue}h_q} , {\color{red}h_j} \wedge{\color{blue}h_q} {\color{gray}\rangle}\label{eq:sfirst}\nonumber\\
&= {\color{gray}\int_{Q_{abde}\cup Q_{bcef}\cup Q_{degh}\cup Q_{efhi}}} {\color{red}h_j(x)}{\color{blue}h_q(y)}{\color{red}h_j(x)}{\color{blue}h_q(y)} {\color{red}dx} {\color{blue}dy}\nonumber\\
&={\color{red}\int_{E_{ijk}} h_j(x) h_j(x) dx} {\color{blue}\int_{E_{pqr}} h_q(y) h_q(y) dy}\nonumber\\
&={\color{red}s_{jj}} {\color{blue}s_{qq}}
\end{align}
\end{fleqn}
\begin{fleqn}
\begin{align}
{s_{ae}} &= {\color{gray}\langle {H_a} , {H_e} \rangle}={\color{gray}\langle} {\color{red}h_i}\wedge {\color{blue}h_r} , {\color{red}h_j} \wedge{\color{blue}h_q} {\color{gray}\rangle}\nonumber\\
&={\color{gray}\int_{Q_{abde}}} {\color{red}h_i(x)} {\color{blue}h_r(y)} {\color{red}h_j(x)} {\color{blue}h_q(y)} {\color{red}dx} {\color{blue}dy}\nonumber\\
&={\color{red}\int_{E_{ij}} h_i(x) h_j(x) dx} {\color{blue}\int_{E_{qr}} h_r(y) h_q(y) dy}\nonumber\\
&={\color{red}s_{ij}} {\color{blue}s_{qr}}
\end{align}
\end{fleqn}
\begin{fleqn}
\begin{align}
{s_{de}} &= {\color{gray}\langle {H_{d}}, {H_e}\rangle}={\color{gray}\langle} {\color{red}h_i} \wedge{\color{blue}h_q} , {\color{red}h_j}\wedge {\color{blue}h_q} {\color{gray}\rangle}\nonumber\\
&={\color{gray}\int_{Q_{abde}\cup Q_{degh}}} {\color{red}h_i(x)} {\color{blue}h_q(y)} {\color{red}h_j(x)} {\color{blue}h_q(y)} {\color{red}dx} {\color{blue}dy}\nonumber\\
&={\color{red}\int_{E_{ij}} h_i(x) h_j(x) dx} {\color{blue}\int_{E_{pqr}} h_q(y) h_q(y) dy}\nonumber\\
&={\color{red}s_{ij}} {\color{blue}s_{qq}}
\end{align}
\end{fleqn}

Similarly, the stiffness integrals read:
\begin{fleqn}
\begin{align}
{w_{ee}} &= {\color{gray}\langle {\nabla H_e}, {\nabla H_e} \rangle}
={\color{gray}\langle \nabla} {\color{red}h_j} \wedge{\color{blue}h_q} , {\color{gray}\nabla} {\color{red}h_j}\wedge {\color{blue}h_q} {\color{gray}\rangle}\nonumber\\
&={\color{gray}\langle} {\color{red}\nabla h_j}{\color{blue}h_q}  , {\color{red}\nabla h_j}{\color{blue}h_q} {\color{gray}\rangle}+2{\color{gray}\langle} {\color{red}h_j}{\color{blue}\nabla h_q} , {\color{red}\nabla h_j}{\color{blue}h_q} {\color{gray}\rangle} + {\color{gray}\langle} {\color{red}h_j}{\color{blue}\nabla h_q} , {\color{red}h_j}{\color{blue}\nabla h_q} {\color{gray}\rangle}\nonumber\\
&={\color{gray}\int_{Q_{abde}\cup Q_{bcef}\cup Q_{degh}\cup Q_{efhi}}} \langle {\color{red}\nabla h_j(x)}{\color{blue}h_q(y)}  , {\color{red}\nabla h_j(x)}{\color{blue}h_q(y)} \rangle {\color{red}dx} {\color{blue}dy} + \cdots \nonumber\\
&={\color{gray}\int_{Q_{abde}\cup Q_{bcef}\cup Q_{degh}\cup Q_{efhi}}} {\color{blue}h_q(y)}{\color{blue}h_q(y)} \langle {\color{red}\nabla h_j(x)} , {\color{red}\nabla h_j(x)} \rangle {\color{red}dx} {\color{blue}dy} + \cdots \nonumber\\
&={\color{red}\int_{E_{ijk}} \langle {\nabla h_j(x)} , {\nabla h_j(x)} \rangle dx}{\color{blue}\int_{E_{pqr}} h_q(y) h_q(y) dy}  + \cdots + \cdots\nonumber\\
&={\color{red}w_{jj}}{\color{blue}s_{qq}}+ {\color{red}s_{jj}}{\color{blue}w_{qq}}
\end{align}
\end{fleqn}
\begin{fleqn}
\begin{align}
{w_{ae}}&= {\color{gray}\langle {\nabla H_a}, {\nabla H_e} \rangle}
={\color{gray}\langle \nabla} {\color{red}h_i} \wedge{\color{blue}h_r} , {\color{gray}\nabla} {\color{red}h_j} \wedge{\color{blue}h_q} {\color{gray}\rangle}\nonumber\\
&={\color{gray}\langle} {\color{red}\nabla h_i}{\color{blue}h_r} , {\color{red}\nabla h_j}{\color{blue}h_q} {\color{gray}\rangle} +{\color{gray}\langle} {\color{red}h_i}{\color{blue}\nabla h_r} ,{\color{red}h_j}{\color{blue}\nabla h_q} {\color{gray}\rangle}\nonumber\\
&={\color{red}w_{ij}}{\color{blue}s_{qr}} + {\color{red}s_{ij}}{\color{blue}w_{qr}}
\end{align}
\end{fleqn}
\begin{fleqn}
\begin{align}
{w_{de}}&= {\color{gray}\langle {\nabla H_d}, {\nabla H_e} \rangle}
={\color{gray}\langle \nabla} {\color{red}h_i} \wedge{\color{blue}h_q} , {\color{gray}\nabla} {\color{red}h_j}\wedge {\color{blue}h_q} {\color{gray}\rangle}\nonumber\\
&={\color{gray}\langle} {\color{red}\nabla h_i}{\color{blue}h_q} , {\color{red}\nabla h_j}{\color{blue}h_q} {\color{gray}\rangle} +{\color{gray}\langle} {\color{red}h_i}{\color{blue}\nabla h_q} ,{\color{red}h_j}{\color{blue}\nabla h_q} {\color{gray}\rangle}\nonumber\\
&={\color{red}w_{ij}}{\color{blue}s_{qq}}+{\color{red}s_{ij}}{\color{blue}w_{qq}}\label{eq:wlast}
\end{align}
\end{fleqn}
where we applied the outer product rule for the gradient operator, and used the fact that $\langle {\color{red}\nabla f},{\color{blue} \nabla g} \rangle = 0$ for any pair of functions on the two cycle graphs. Note the integrals $s_{ae}$ and $w_{ae}$ are non-zero even if nodes $a$ and $e$ are not connected in the product graph.

In matrix notation, formulas \eqref{eq:sfirst}-\eqref{eq:wlast} can be succinctly written as:
\begin{align}
{\color{gray}\mathbf{S}} &= {\color{red}\mathbf{S}}\otimes{\color{blue}\mathbf{S}}\nonumber\\
{\color{gray}\mathbf{W}} &= {\color{red}\mathbf{W}}\otimes{\color{blue}\mathbf{S}} + {\color{red}\mathbf{S}}\otimes{\color{blue}\mathbf{W}}\,,\nonumber
\end{align}
completing the proof. $\square$

\begin{figure}[tb]
\centering
\begin{overpic}
[trim=0cm 0cm 0cm 0cm,clip,width=0.35\linewidth]{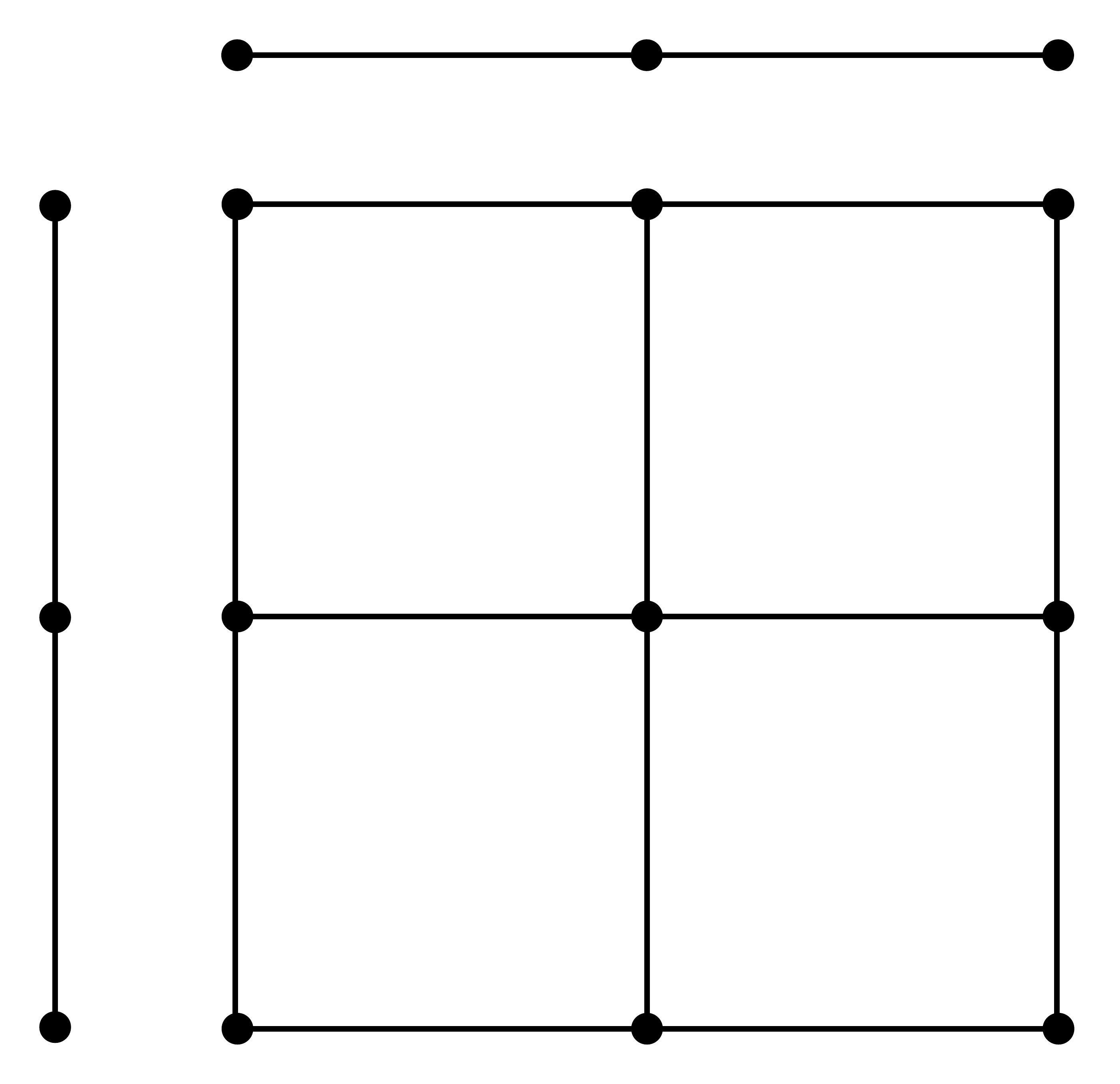}
\put(18,80.5){\color{gray}$a$}
\put(53,80.5){\color{gray}$b$}
\put(93,80.5){\color{gray}$c$}
\put(18,44){\color{gray}$d$}
\put(53,44){\color{gray}$e$}
\put(93,44){\color{gray}$f$}
\put(18,-4){\color{gray}$g$}
\put(53,-4){\color{gray}$h$}
\put(93,-4){\color{gray}$i$}
\put(-4, 42){\color{blue}$q$}
\put(-4, 78){\color{blue}$r$}
\put(-4, 3){\color{blue}$p$}
\put(18, 96){\color{red}$i$}
\put(53, 96){\color{red}$j$}
\put(93, 96){\color{red}$k$}
\put(97, 88.5){$\cdots$}
\put(7, 88.5){$\cdots$}
\put(4, 81){\footnotesize $\vdots$}
\put(4, -8){\footnotesize $\vdots$}
\end{overpic}
\hspace{0.45cm}
\begin{overpic}
[trim=0cm 0cm 0cm 0cm,clip,width=0.57\linewidth]{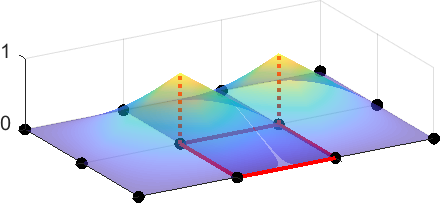}
\put(38,31){$H_e$}
\put(60,36){$H_f$}
\put(39,8){$e$}
\put(53,0){$h$}
\put(74,4){$i$}
\put(62,11){$f$}
\end{overpic}
\caption{\label{fig:quad}{\em Left}: The product of two closed contours discretized as cycle graphs (in blue and red) is a quad mesh with toric topology (in grey). Uniform edge lengths are used for illustration purposes. {\em Right}: Two overlapping bilinear hats $H_e$ and $H_f$. On the quad element $Q_{efhi}$ (marked in red) there is non-zero overlap, hence it contributes to the computation of mass and stiffness values.}
\end{figure}
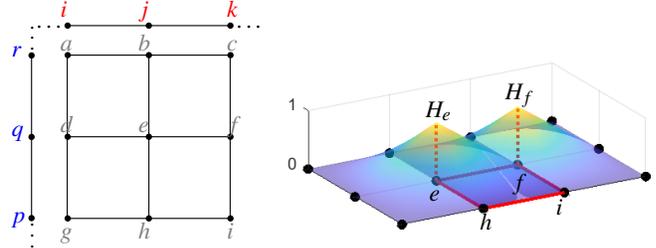

\noindent\textbf{Proof of Corollary~\ref{thm:L}.}
The proof is straightforward and follows from substituting the expressions \eqref{eq:sprod}, \eqref{eq:wprod} into the general formula $\L=\S^{-1}\W$:
\begin{align}
\mathbf{L}_{\M\times\N} &= {\mathbf{S}}^{-1}_{\M\times\N} {\mathbf{W}_{\M\times\N}}\nonumber\\
&=(\S_\M\otimes\S_\N)^{-1}( \W_\M\otimes\S_\N + \S_\M\otimes\W_\N)\nonumber\\
&=(\S_\M^{-1}\otimes\S_\N^{-1})( \W_\M\otimes\S_\N)+(\S_\M^{-1}\otimes\S_\N^{-1})( \S_\M\otimes\W_\N)\nonumber\\
&=(\S_\M^{-1}\W_\M)\otimes(\S_\N^{-1}\S_\N)+(\S_\M^{-1}\S_\M)\otimes(\S_\N^{-1}\W_\N)\nonumber\\
&={\mathbf{L}_\M}\otimes\mathbf{I}_\N + \mathbf{I}_\M\otimes{\mathbf{L}_\N}\,.
\hspace{3.6cm}\square\nonumber
\end{align}

\noindent\textbf{Proof of Corollary~\ref{thm:LB2D}.}
Since triangular (3-3) duoprisms are, by definition, the Cartesian product of two triangles, we can define a multilinear basis function on the product complex as the outer product of two standard hats defined on triangle meshes. We are now in the same setting as the lower dimensional case, and in particular Equations~\eqref{eq:sfirst}-\eqref{eq:wlast} remain valid. $\square$

\bibliographystyle{eg-alpha-doi}
\bibliography{egbib}

\end{document}